\documentclass[12pt,reqno]{amsart}

\usepackage[left=1in, right=1in, top=1.1in,bottom=1.1in]{geometry}
\usepackage[title]{appendix}

\usepackage{graphicx}
\usepackage{subcaption}
\graphicspath{ {./fig/}}
\usepackage[nomarkers, nolists, figuresonly]{endfloat}
\usepackage{caption}

\usepackage[T1]{fontenc}
\usepackage{amsthm, amssymb, amsmath, amsfonts, mathrsfs, amscd}
\usepackage{enumitem}
\usepackage{stmaryrd}
\usepackage{mathrsfs}
\usepackage{bbm}
\usepackage{xcolor}
\usepackage{booktabs}

\usepackage{pdfsync}
\usepackage[font={scriptsize}]{caption}
\usepackage{verbatim}

\usepackage{hyperref}
\hypersetup{
  colorlinks   = true,
  citecolor    = green,
  linkcolor    = red
}

\newtheorem{theorem}{Theorem}[section]

\newtheorem{corollary}[theorem]{Corollary}

\newtheorem{proposition}[theorem]{Proposition}

\theoremstyle{remark}

\theoremstyle{remark}
\newtheorem{example}[theorem]{Example}

\newcommand{\p}{\partial}

\newcommand{\1}{\mathbbm{1}}

\title[Pricing and hedging for liquidity provision in CFMM]{Pricing and hedging for liquidity provision in Constant Function Market Making}
\author[J. Risk]{Jimmy Risk}
\author[S.-N. Tung]{Shen-Ning Tung}
\author[T.-H. Wang]{Tai-Ho Wang}
\date{\today}

\address{Jimmy Risk \newline
Department of Mathematics and Statistics, \newline
Cal Poly Pomona\newline
3801 W Temple Ave, Pomona CA 91768
}
\email{jrisk@cpp.edu}

\address{Shen-Ning Tung \newline
Department of Mathematics, \newline
National Tsing Hua University \newline
Hsinchu, Taiwan
}
\email{tung@math.nthu.edu.tw}

\address{Tai-Ho Wang \newline
Department of Mathematics \newline
Baruch College, The City University of New York \newline
1 Bernard Baruch Way, New York, NY10010
}
\email{tai-ho.wang@baruch.cuny.edu}


\begin{document}

\begin{abstract}
This paper develops a robust mathematical framework for Constant Function Market Makers (CFMMs) by transitioning from traditional token reserve analyses to a coordinate system defined by price and intrinsic liquidity. We establish a canonical parametrization of the bonding curve that ensures dimensional consistency across diverse trading functions, such as those employed by Uniswap and Balancer, and demonstrate that asset reserves and value functions exhibit a linear dependence on this intrinsic liquidity. This linear structure facilitates a streamlined approach to arbitrage-free pricing, delta hedging, and systematic risk management. By leveraging the Carr-Madan spanning formula, we characterize Impermanent Loss (IL) as a weighted strip of vanilla options, thereby defining a fine-grained implied volatility structure for liquidity profiles. Furthermore, we provide a path-dependent analysis of IL using the last-passage time. Empirical results from Uniswap v3 ETH/USDC pools and Deribit option markets confirm a volatility smile consistent with crypto-asset dynamics, validating the framework's utility in characterizing the risk-neutral fair value of liquidity provision.
\end{abstract}

\maketitle

\begin{center}
{\it It all begins with the canonical parametrization of the bonding curve in CFMM.}
\end{center}

\allowdisplaybreaks  

\section{Introduction}
The emergence of \textit{Decentralized Finance} (DeFi) \cite{harvey2021defi, gobet2023decentralized, capponi2021adoption} has revolutionized asset exchange through the implementation of \textit{Constant Function Market Makers} (CFMMs) \cite{angeris2020improved, angeris2024geometry}. Unlike traditional limit order books that rely on discrete orders, protocols such as \textit{Uniswap} \cite{adams2020uniswapv2, adams2021uniswapv3, adams2023uniswapv4} and \textit{Balancer} \cite{martinelli2019balancer} utilize a deterministic bonding curve to govern the relationship between token reserves and exchange rates. This mechanism ensures that liquidity remains available across a continuous price spectrum, contingent on the slippage costs defined by the underlying trading function.

Within this ecosystem, the role of the \textit{Liquidity Provider} (LP) has undergone a profound transformation. By depositing assets into a smart contract, LPs effectively sell tokens to the market according to a pre-defined rule \cite{angeris2023replicating}. Consequently, an LP position can be mathematically framed as a sophisticated financial derivative. The growing academic and industrial recognition of this equivalence is highlighted by protocols such as Panoptic \cite{lambert2022panoptic}, which utilize Uniswap v3 liquidity positions as the underlying engine for perpetual options.

As the complexity of these automated strategies increases, research has shifted toward the \textit{optimal design of bonding curves}. While early automated market makers (AMMs) relied on simple invariants, the current landscape demands capital-efficient structures that reflect specific financial payoffs. Determining an ``optimal'' curve requires a rigorous understanding of how liquidity profiles interact with market volatility and informed order flow.

\subsection*{Literature Review}

\subsubsection*{LP Positions as Synthetic Derivatives}
A foundational branch of literature characterizes the payoff structures of LPs as equivalent to traditional derivatives. Early work replicated Uniswap v2 payoffs with portfolios of options, a framework subsequently extended to the concentrated liquidity models of Uniswap v3 \cite{clark2020replicating, clark2021replicating}. This equivalence was formalized by Angeris et al. \cite{angeris2023replicating}, who characterized the space of feasible CFMM payoffs in terms of the properties of the trading function. Within this lens, \textit{Impermanent Loss} (IL) is increasingly viewed as a hedgeable risk; notably, Fukasawa et al. \cite{fukasawa2023weighted} demonstrate that IL in constant product markets can be perfectly hedged using a weighted variance swap of order $1/2$. More practical hedging strategies involve standard put and call options to mitigate the idiosyncratic risks of concentrated liquidity positions \cite{lipton2025unified}.

\subsubsection*{Pricing, Hedging, and Implied Volatility}
As LP positions are recognized as derivatives, the question of their ``fair value'' arises. Recent works establish a rigorous framework for the arbitrage-free pricing of liquidity tokens \cite{bichuch2025arbitrage, bichuch2025price}. By defining an implied volatility metric derived from transaction-fee streams, this research enables LPs to quantify the ``price of liquidity'' in a manner consistent with the Black--Scholes paradigm, bridging the gap between endogenous fee generation and exogenous asset volatility.

\subsubsection*{Dynamic Modeling and Parameterization}
To move beyond static payoffs, recent studies model the temporal evolution of liquidity. The framework introduced in \cite{tung2024mathematical} provides a canonical parameterization of CFMMs through a ``liquidity profile,'' offering a unified mathematical language for various bonding curves. Building on this, \cite{risk2025dynamics} models the stochastic dynamics of these profiles to capture shifts in response to market behavior. These models are essential for analyzing the ``Loss-Versus-Rebalancing'' (LVR) metric, which \cite{singh2025modeling} presents as a continuous-installment option to price the ongoing costs of providing liquidity against informed arbitrageurs.

\subsubsection*{Strategic and Optimal Liquidity Provision}
Finally, a growing body of work focuses on LPs' strategic behavior. Using an optimal stopping approach, Capponi and Zhu \cite{capponi2025optimal} demonstrate the existence of strategies that remain profitable relative to ``buy-and-hold'' benchmarks, even after accounting for LVR \cite{capponi2025optimal}. This suggests that the ``natural coordinates'' for AMM analysis are price and intrinsic liquidity, which together dictate the optimal design of the bonding function.

\subsection*{Contribution}
This paper reviews and extends the market model framework proposed in \cite{tung2024mathematical}. By shifting the perspective from token reserves to a coordinate system defined by price and liquidity, the framework achieves several key objectives:
\begin{enumerate}
    \item It defines an intrinsic liquidity that maintains dimensional consistency ($\sqrt{{\rm ETH} \times {\rm USDC}}$) across all CFMM bonding functions, providing a universal metric for market depth.
    \item It demonstrates that relevant financial quantities, including asset reserves and value functions, are encoded in the linear dependence of the model on the intrinsic liquidity profile.
    \item It simplifies complex tasks such as \textit{arbitrage-free pricing}, \textit{delta hedging}, and \textit{systematic risk management} by exploiting this linear structure.
    \item It unifies diverse results from existing literature---ranging from constant product formulas to concentrated liquidity---into a single, concise representation.
    \item It enables the evaluation of liquidity profiles through fine-structure implied volatility and path-dependent last-passage time analysis, allowing for a granular assessment of market-implied risk.
\end{enumerate}

We argue that the \textit{price-intrinsic liquidity pair} is the natural coordinate system for the pricing and hedging of bonding functions. In the following sections, we detail the mathematical derivation of this framework and demonstrate its application to several canonical AMM structures.
\section{Intrinsic Liquidity and Liquidity Profile} \label{sec:Uniswap}

In the study of Constant Function Market Makers (CFMMs) \cite{angeris2020improved, angeris2024geometry}, the system is typically characterized by a bonding function $f$ and its associated level set $f(x,y) = K$, often termed the \textit{bonding curve}. Conventionally, the level $K$ is interpreted as the ``liquidity depth.'' We argue, however, that this definition is \textit{dimensionally inconsistent} across different CFMM architectures. To resolve this, we propose an \textit{intrinsic} definition of liquidity that is invariant under reparametrization and maintains dimensional consistency across all CFMM models. While aspects of this discussion appear in \cite{tung2024mathematical}, we provide the key arguments here for completeness.

\subsection{Dimensionally Consistent Liquidity}
Consider the Constant Product Market Maker (CPMM) \cite{adams2020uniswapv2} for a pair of tokens, e.g., ETH and USDC. The bonding curve is defined by the level set $\sqrt{xy} = K$, where the parameter $K$ possesses the physical dimension $\sqrt{\mathrm{ETH} \times \mathrm{USDC}}$. In contrast, for Geometric Mean Market Makers (G3Ms) \cite{martinelli2019balancer}, the bonding curve is given by $x^\alpha y^{1-\alpha} = K$ for some weight $\alpha \in (0,1)$. For this model, $K$ carries the dimension $\mathrm{ETH}^\alpha \times \mathrm{USDC}^{1-\alpha}$.

This dimensional disparity implies that the parameter $K$ is not a consistent proxy for market depth across protocols; specifically, the same numerical value of $K$ represents different physical quantities depending on the exponent $\alpha$. To resolve this, we utilize the framework established in \cite[Appendix B]{tung2024mathematical} to define a \textit{local intrinsic liquidity} $\ell$ that is both dimensionally consistent and invariant under reparametrization. For a smooth bonding curve $f(x,y)=K$, we define $\ell$ at the reserve state $(x,y)$ as:
\begin{equation}
\ell := \frac{-2(f_x f_y)^{3/2}}{f_{yy}f_x^2 - 2f_{xy}f_xf_y + f_{xx}f_y^2}, \label{eqn:ell}
\end{equation}
where subindices denote partial derivatives.

The definition in \eqref{eqn:ell} ensures that $\ell$ always yields the dimension $\sqrt{\mathrm{ETH} \times \mathrm{USDC}}$, matching the CPMM standard regardless of the functional form of $f$. Mathematically, this definition is related to the \textit{curvature} of the bonding curve in the $(x,y)$ plane. We emphasize two critical properties of $\ell$:
\begin{itemize}
    \item \textit{Locality:} Unlike the global constant $K$, $\ell$ is a state-dependent quantity $\ell(x,y)$ that characterizes the immediate price impact at a specific reserve point.
    \item \textit{Invariance:} The definition is intrinsic to the geometry of the curve $f(x,y)=K$, meaning it is independent of the choice of functional representation (e.g., $xy=K^2$ versus $\sqrt{xy}=K$).
\end{itemize}

\begin{example}[G3M Liquidity]
For the G3M $f(x,y) = x^\alpha y^{1 - \alpha} = K$, the intrinsic liquidity is computed as:
\begin{equation*}
\ell = 2 \sqrt{\alpha (1 - \alpha)}\sqrt{xy} = 2 \sqrt{\alpha (1 - \alpha)} K \, x^{\frac{1}{2}-\alpha} y^{\alpha - \frac{1}{2}}. \label{cite: 75}
\end{equation*}
This expression explicitly restores the $\sqrt{\mathrm{ETH} \times \mathrm{USDC}}$ dimension. In the specific case of the CPMM ($\alpha = \frac{1}{2}$), the local intrinsic liquidity simplifies to $\ell = \sqrt{xy} = K$. Thus, for a CPMM, the liquidity profile is uniform across all price points, and the intrinsic liquidity coincides with the traditional depth parameter.
\end{example}

\subsection{Liquidity Profile and Reserves}
In the study of CFMMs, the spot price $p$ is endogenously determined by the pool's current reserve state $(x,y)$. By adopting token $Y$ as the num\'{e}raire, the spot price is defined as the marginal rate of substitution:
\begin{equation*}
p := -\frac{dy}{dx} = \frac{f_x}{f_y}.
\end{equation*}
This value represents the unit price of token $X$ required for an infinitesimal swap $dx$. While conventional models define reserves $x$ and $y$ as primary variables, we demonstrate that the bonding curve $f(x,y)=K$ admits a more fundamental representation via the spot price $p$ and the local intrinsic liquidity $\ell$. This leads to what we term the \textit{canonical parametrization} of the bonding curve.

\begin{theorem}[Canonical Parametrization] \label{thm:canonical}
Let the bonding function $f(x,y)$ be smooth, strictly increasing, and convex in $(x,y)$. For the CFMM $f(x,y)=K$, the reserve levels $(x,y)$ are uniquely determined by the spot price $p$ and the local intrinsic liquidity $\ell$ via the following integral representations:
\begin{equation}
x(p) = \int_p^\infty \frac{\ell(q)}{2\sqrt{q^3}} \, dq, \qquad
y(p) = \int_0^p \frac{\ell(q)}{2\sqrt q} \, dq, \label{eqn:canon-param}
\end{equation}
where $\ell(q)$ is the intrinsic liquidity at price $q$ defined in \eqref{eqn:ell}.
\end{theorem}

\begin{proof}
Consider an arbitrary parametrization $(x(s), y(s))$ of the level curve $f(x,y)=K$ with $\dot{x} < 0$ and $\dot{y} > 0$. Implicit differentiation yields $f_x \dot{x} + f_y \dot{y} = 0$, implying $p = -\dot{y}/\dot{x}$. By differentiating $p$ with respect to the parameter $s$, we have:
\begin{eqnarray*}
\frac{dp}{ds} &=& -\frac{d}{ds}\left(\frac{\dot y}{\dot x}\right) 
= \frac{-\ddot y\dot x + \ddot x \dot y}{\dot x^2} \\
&=& \frac{d}{ds}\left(\frac{f_x}{f_y}\right) = -\frac{\dot y}{f_x f_y^2}(f_{xx}f_y^2 - 2f_{xy}f_xf_y + f_{yy}f_x^2).
\end{eqnarray*}
Utilizing the definition of $\ell$ in \eqref{eqn:ell} and changing the variable of integration from $s$ to $p$, we obtain:
\begin{eqnarray*}
\ell &=& \frac{-2(f_x f_y)^{3/2}}{f_{yy}f_x^2 - 2f_{xy}f_xf_y + f_{xx}f_y^2}
= 2 \, \frac{(-\dot x\dot y)^{3/2}}{\ddot y\dot x - \ddot x\dot y}, \\
\frac{dx}{dp} &=& \dot x \frac{ds}{dp} = -\frac{\dot x^3}{\ddot y\dot x - \ddot x\dot y} = - \frac12 \frac{2(-\dot x\dot y)^{3/2}}{\ddot y\dot x - \ddot x\dot y} \left(-\frac{\dot x}{\dot y}\right)^{3/2} = -\frac{\ell}{2 p^{3/2}}, \\
\frac{dy}{dp} &=& \dot y \frac{ds}{dp} = -\frac{\dot x^2 \dot y}{\ddot y\dot x - \ddot x\dot y} = \frac12 \frac{2(-\dot x\dot y)^{3/2}}{\ddot y\dot x - \ddot x\dot y} \sqrt{-\frac{\dot x}{\dot y}} = \frac{\ell}{2 \sqrt p}.
\end{eqnarray*}
Integrating these expressions with respect to $p$ from the respective boundaries ($p \to \infty$ for $x$ and $p \to 0$ for $y$) completes the proof.
\end{proof}

The relationship in Theorem \ref{thm:canonical} allows us to define the \textit{local intrinsic liquidity profile} as a distribution function $\ell(p)$ over the price spectrum $(0, \infty)$. For the purpose of financial modeling, it is often more convenient to work with the \textit{liquidity profile} $L(q)$, defined as:
\begin{equation}
L(q) := \frac{\ell(q)}{2q^{3/2}}.
\end{equation}
Substituting $L(q)$ into \eqref{eqn:canon-param} yields a succinct representation of the pool reserves:
\begin{eqnarray} \label{eq:reserves}
&& x(p) = \int_p^\infty \frac{\ell(q)}{2\sqrt{q^3}} dq = \int_p^\infty L(q) dq , \quad y(p) = \int_0^p \frac{\ell(q)}{2\sqrt q} dq = \int_0^p q L(q) dq. %\label{eq:reserves_x}
\end{eqnarray}
The dimensional consistency is preserved: $L(q)dq$ carries the dimension of token $X$, while $qL(q)dq$ carries the dimension of token $Y$. This formulation highlights the \textit{infinitesimal price impact} of a trade volume $dx$:
\begin{equation}
dp = -\frac{1}{L(p)}dx.
\end{equation}
This identity confirms the economic intuition that higher liquidity density $L(p)$ suppresses slippage, as a larger $L(p)$ results in a smaller price adjustment $dp$ for any given volume $dx$.

\subsection{Liquidity Profile Value and Option Replication} \label{sec:pool-value}
For a given liquidity profile $L$, we define its mark-to-market value $V_L(p)$ as a function of the spot price $p$. Taking token $Y$ as the num\'{e}raire, the value is defined by the sum of the value of its constituent token holdings:
\begin{align}
V_L(p) &= x(p) p + y(p) \notag \\
&= \int_0^\infty p \1_{[p,\infty)}(q) L(q) dq + \int_0^\infty q \1_{[0, p]}(q) L(q) dq \notag \\
&= \int_0^\infty \min\{p, q\} L(q) dq, \label{eq:pool_value_min}
\end{align}
where $\1_A(\cdot)$ denotes the indicator function for the interval $A$. The representation in \eqref{eq:pool_value_min} provides a direct financial interpretation: since $\min\{p, q\} = p - (p-q)^+$ corresponds to the payoff of a covered call struck at $q$, the total pool value $V_L$ is equivalent to the payoff of a \textit{strip of weighted covered calls}. The weights of this portfolio are determined entirely by the liquidity profile $L(q)$.

A more fundamental decomposition is achieved by applying the spanning formula of Carr and Madan \cite{carr2001towards}, which enables the representation of any twice-differentiable payoff $\varphi(S)$ via a portfolio of OTM options:
\begin{equation}
\varphi(S) = \varphi(S_0) + \varphi'(S_0)(S - S_0) + \int_0^{S_0} \varphi''(K)(K - S)^+ \, dK + \int_{S_0}^\infty \varphi''(K) (S - K)^+ \, dK. \label{eqn:spanning}
\end{equation}
By noting that $V_L'(p) = x(p)$ and $V_L''(p) = x'(p) = -L(p)$, we apply this decomposition to the value function $V_L$. Substituting these derivatives into \eqref{eqn:spanning} yields:
\begin{eqnarray} \label{eqn:LP_value}
V_L(p) &=& V_L(p_0) + V_L'(p_0)(p - p_0) + \int_0^{p_0} V_L''(q)(q - p)^+ dq + \int_{p_0}^\infty V_L''(q) (p - q)^+ dq \notag \\
&=& x(p_0)p + y(p_0) - \int_0^{p_0} L(q) (q - P)^+ dq - \int_{p_0}^\infty L(q) (p - q)^+ dq. \label{eq:carr_madan}
\end{eqnarray}
The identity in \eqref{eq:carr_madan} demonstrates that the liquidity provider's position is equivalent to a linear combination of the initial holding value and a \textit{short position} in a continuously weighted strip of out-of-the-money (OTM) put options (for $q \in [0, p_0]$) and OTM call options (for $q \in [p_0, \infty)$). This equivalence is central to the pricing and hedging strategies developed in subsequent sections. In a complete market, \eqref{eq:carr_madan} implies that the value $V_L$ can be perfectly replicated by a self-financing trading strategy involving the underlying assets and vanilla options.

\subsection{Replicating Market Makers} \label{sec:RMM}
In \cite{angeris2023replicating}, it is established that every CFMM possesses a concave, nonnegative, and 1-homogeneous payoff function, and conversely, any payoff function satisfying these properties has a corresponding convex CFMM. We demonstrate in this section that the canonical relationships established in \eqref{eq:reserves} provide a straightforward, alternative approach for constructing a CFMM trading function for any given target portfolio value $V(p)$.

Assume a liquidity provider (LP) wishes to replicate a target portfolio value $V(p) := p x(p) + y(p)$ by designing an appropriate CFMM. Recall from the canonical parametrization that $x(p) = V'(p)$. By inverting this relationship, we obtain the spot price as a function of reserves: $p = (V')^{-1}(x)$. Furthermore, since the liquidity profile is given by $L(p) = -V''(p)$, we can derive the CFMM trading function by integrating the reserve relationship for $y$:
\begin{equation}
y + \int_0^{(V')^{-1}(x)} q V''(q) \, dq = 0. \label{eqn:RMM_general}
\end{equation}
We illustrate this procedure with two canonical examples from the literature.

\begin{example}[Payoff of a Covered Call at Expiry]
Consider an LP whose target portfolio value is the payoff of a covered call struck at $K$ at expiry:
\begin{equation*}
V(p) = p - (p - K)^+.
\end{equation*}
The corresponding liquidity profile is a Dirac delta function: $L(p) = -V''(p) = \delta(p - K)$. Applying the canonical parametrization yields:
\begin{align*}
x(p) &= \int_p^\infty \delta(q - K) \, dq = \theta(K - p), \\
y(p) &= \int_0^p q \, \delta(q - K) \, dq = K \theta(p - K),
\end{align*}
where $\theta$ denotes the Heaviside step function. Using the identity $1 - \theta(K - p) = \theta(p - K)$, we recover the linear trading function:
\begin{equation*}
Kx + y = K,
\end{equation*}
consistent with the results in \cite{angeris2023replicating}.
\end{example}

\begin{example}[Black--Scholes Covered Call]
Assume the target portfolio value is the price of a covered call under Black--Scholes dynamics:
\begin{equation*}
V(p) = p - p N(d_1) + K N(d_2),
\end{equation*}
where $d_1 = \frac{\ln(p/K)}{v} + \frac{v}{2}$ and $d_2 = d_1 - v$. Since $x = V'(p) = 1 - N(d_1)$, we find the price as a function of the $x$ reserve: $p = K e^{vN^{-1}(1-x) - v^2/2}$. Given that $L(p) = -V''(p) = \Gamma(p)$ (the Black--Scholes Gamma), the $y$ reserve is obtained via:
\begin{eqnarray*}
y &=& \int_0^p q L(q) dq
= \int_0^{K e^{vN^{-1}(1 - x) - \frac{v^2}2}} q \, \Gamma(q) dq 
= K N\left(N^{-1}(1 - x) - v\right),
\end{eqnarray*}
This can be expressed in the symmetric form:
\begin{equation*}
N^{-1}(1 - x) - N^{-1}(y/K) = v, \label{cite: 185}
\end{equation*}
recovering the replicating market maker results for perpetual options in \cite{angeris2023replicating}.
\end{example}
\section{Pricing and Hedging for General Liquidity Profiles} \label{sec:Pricing}

In this section, we establish that liquidity provision within a CFMM pool is mathematically equivalent to maintaining a short position in a contingent claim characterized by a convex payoff. Leveraging the canonical parametrization developed in \eqref{eq:reserves}, we quantify the risk inherent to the liquidity provider (LP) and derive the associated risk-neutral pricing and hedging frameworks by exploiting the spanning formula introduced in \eqref{eqn:LP_value}.

\subsection{Static Replication of Impermanent Loss by Vanilla Options} \label{sec:IL-options}
The primary risk for an LP in a CFMM is quantified as \textit{impermanent loss} (IL), also referred to as \textit{divergence loss}. Defined as the opportunity cost of providing liquidity relative to a simple buy-and-hold strategy of the underlying assets, IL can be expressed as a weighted sum of payoffs from a continuous strip of vanilla options.

Let $p_0$ denote the current pool price and $p_T$ the future price at time $T$. Utilizing the expression for the value function $V_L(p_T)$ from \eqref{eq:carr_madan}, the IL at time $T$ is defined as:
\begin{align}
{\rm IL}(p_T|p_0, L) &:= x(p_0)p_T + y(p_0) - V_L(p_T) \notag \\
&= \int_0^{p_0} L(q) (q - p_T)^+ dq + \int_{p_0}^\infty L(q) (p_T - q)^+ dq. \label{eqn:IL_replication}
\end{align}
This representation highlights the \textit{linearity of IL with respect to the liquidity profile $L$} and its universal applicability across diverse bonding functions. The formula can be simplified to a more concise integral form:
\begin{equation}
{\rm IL}(p_T|p_0, L) = \int_{p_0}^{p_T} (p_T - q) L(q) dq, \label{eqn:IL-short}
\end{equation}
which underscores that IL is essentially the cost of the option embedded in the LP's position rather than a mere metric of pool stability.

This replication framework unifies several disparate expressions found in the literature \cite{clark2020replicating, clark2021replicating, maire2024market, lipton2025unified}. For instance, \cite{lipton2025unified} showed that a concentrated liquidity position in the range $[p_a, p_b]$ with a single unit of liquidity ($L(q) = \frac{1}{2q^{3/2}}\1_{\{p_a \leq q \leq p_b\}}$) can be decomposed into terms $u_0$, $u_{1/2}$, and $u_1$ as follows:
\begin{align*}
u_0(p_t) &= \sqrt{p_0} + \frac{p_t}{\sqrt{p_0}}, \\
u_{1/2}(p_t) &= - 2\sqrt{p_t} \1_{\{p_a \leq p_t \leq p_b\}}, \\
u_1(p_t) &= -2\sqrt{p_a} \1_{\{p_t < p_a\}} + \frac{1}{\sqrt{p_a}}(p_a - p_t)^+ - 2\sqrt{p_b} \1_{\{p_t > p_b\}} - \frac{1}{\sqrt{p_b}}(p_t - p_b)^+.
\end{align*}
Applying the general result in \eqref{eqn:IL_replication} provides a direct derivation of this specific case. By integrating over the range $[p_a, p_b]$ and considering the price relative to the boundaries, we recover the tripartite formula:
\begin{eqnarray*}
{\rm IL}(p_t|p_0, L) &=& \int_{p_a}^{p_0} (q - p_t)^+ \frac{dq}{2q^{3/2}} + \int_{p_0}^{p_b} (p_t - q)^+ \frac{dq}{2q^{3/2}} \\
&=& \left\{\int_{p_a}^{p_0} (q - p_t)^+ \frac{dq}{2q^{3/2}} + \int_{p_0}^{p_b} (p_t - q)^+ \frac{dq}{2q^{3/2}}\right\} \1_{\{p_t < p_a\}} \\
&& + \left\{\int_{p_a}^{p_0} (q - p_t)^+ \frac{dq}{2q^{3/2}} + \int_{p_0}^{p_b} (p_t - q)^+ \frac{dq}{2q^{3/2}}\right\} \1_{\{p_t > p_b\}} \\
&& + \left\{\int_{p_a}^{p_0} (q - p_t)^+ \frac{dq}{2q^{3/2}} + \int_{p_0}^{p_b} (p_t - q)^+ \frac{dq}{2q^{3/2}}\right\} \1_{\{p_a \leq p_t \leq p_b\}} \\
&=& \left\{\sqrt{p_0} - \sqrt{p_a} - \frac{p_t}{\sqrt{p_a}} + \frac{p_t}{\sqrt{p_0}} \right\} \1_{\{p_t < p_a\}} \\
&& + \left\{\frac{p_t}{\sqrt{p_0}} - \frac{p_t}{\sqrt{p_b}} - \sqrt{p_b} + \sqrt{p_0} \right\} \1_{\{p_t > p_b\}} \\
&& + \left\{\sqrt{p_0} - 2\sqrt{p_t} + \frac{p_t}{\sqrt{p_0}} \right\} \1_{\{p_a \leq p_t \leq p_b\}} \\
&=& \left(\sqrt{p_0} + \frac{p_t}{\sqrt{p_0}}\right) - 2\sqrt{p_t} \1_{\{p_a \leq p_t \leq p_b\}} - \left\{ \sqrt{p_a} + \frac{p_t}{\sqrt{p_a}} \right\} \1_{\{p_t < p_a\}} - \left\{ \frac{p_t}{\sqrt{p_b}} + \sqrt{p_b} \right\} \1_{\{p_t > p_b\}} \\
&=& \left(\sqrt{p_0} + \frac{p_t}{\sqrt{p_0}}\right) - 2\sqrt{p_t} \1_{\{p_a \leq p_t \leq p_b\}} \\
&& - \left[ 2\sqrt{p_a} - \frac{1}{\sqrt{p_a}}(p_a - p_t) \right] \1_{\{p_t<p_a\}} - \left[ 2\sqrt{p_b} + \frac{1}{\sqrt{p_b}}(p_t - p_b) \right] \1_{\{p_t>p_b\}} \\
&=& \sqrt{p_0} + \frac{p_t}{\sqrt{p_0}} - 2\sqrt{p_t} \1_{\{p_a \leq p_t \leq p_b\}} \\
&& -2\sqrt{p_a} \1_{\{p_t<p_a\}} + \frac{1}{\sqrt{p_a}}(p_a - p_t)^+ - 2\sqrt{p_b} \1_{\{p_t>p_b\}} - \frac{1}{\sqrt{p_b}}(p_t - p_b)^+ \\
&=& u_0(p_t) + u_{1/2}u_0(p_t) + u_1 u_0(p_t).
\end{eqnarray*}

\subsection{Dynamic Decomposition: Loss-Versus-Rebalancing}
In this subsection, we transition to a continuous-time framework to analyze the temporal evolution of Impermanent Loss (IL). A critical, non-hedgeable component of IL is the \textit{loss-versus-rebalancing} (LVR), first introduced in \cite{milionis2022loss-versus-rebalancing}. We utilize the canonical parametrization in \eqref{eq:reserves} to provide alternative understandings of the relationship between IL, LVR, and contracts on realized variance such as variance and gamma swaps.

\subsubsection{IL as Hedging Cost plus LVR}
By applying the Tanaka formula, we derive an alternative expression for the LVR in terms of the local times of the underlying price process $P_t$. Given that $\mathrm{IL}(P_0 | P_0, L) = 0$, the evolution of the impermanent loss can be written as:
\begin{align}
\mathrm{IL}(P_t | P_0, L) &= \int_0^t d\left( \mathrm{IL}(P_s | P_0, L) \right) \notag \\
&= \int_0^t \int_0^{P_0} L(q) d(q - P_s)^+ dq + \int_0^t \int_{P_0}^\infty L(q) d(P_s - q)^+ dq. \label{eqn:IL_change_tanaka_start}
\end{align}
Applying the Tanaka formula:
\begin{equation*}
d(P_t - q)^+ = -\1_{\{P_t \leq q\}} dP_t + \frac{1}{2} d\mathcal{L}_t^q(P), \quad d(q - P_t)^+ = \1_{\{P_t \geq q\}} dP_t + \frac{1}{2} d\mathcal{L}_t^q(P), 
\end{equation*}
where $\mathcal{L}_t^q(P)$ denotes the local time of $P$ at $q$, we obtain:
\begin{eqnarray}
&& {\rm IL}(P_t|P_0, L) \nonumber \\
&=& \int_0^{P_0} L(q) \left\{ -\int_0^t \1_{\{P_s < q\}} dP_s + \frac{1}{2} \mathcal{L}_t^q(P) \right\} dq + \int_{P_0}^\infty L(q) \left\{ \int_0^t \1_{\{P_s > q\}} dP_s + \frac{1}{2} \mathcal{L}_t^q(P) \right\} dq \notag \\
&=& -\int_0^t \left\{ \int_0^{P_0} L(q) \1_{\{P_s < q\}} dq - \int_{P_0}^\infty L(q) \1_{\{P_s > q\}} dq \right\} dP_s + \frac{1}{2} \int_0^\infty L(q) \mathcal{L}_t^q(P) dq. \qquad \label{eqn:IL_local_time_full}
\end{eqnarray}
Recalling the reserve identity $x(p) = \int_p^\infty L(q) \, dq$ , the term multiplying $dP_s$ is identified as the change in token $X$ reserves relative to the initial state:
\begin{equation*}
\int_0^{P_0} L(q) \1_{\{P_s < q\}} dq - \int_{P_0}^\infty L(q) \1_{\{P_s > q\}} dq = x(P_s) - x(P_0).
\end{equation*}
Substituting this into \eqref{eqn:IL_local_time_full} yields the fundamental dynamic decomposition:
\begin{equation}
\mathrm{IL}(P_t | P_0, L) = \int_0^t \left\{x(P_0) - x(P_s) \right\} dP_s + \frac{1}{2} \int_0^\infty \mathcal{L}_t^q(P) L(q) dq. \label{eq:IL_decomposition}
\end{equation}
The first integral, $\int_0^t \{x(P_0) - x(P_s)\} dP_s$, represents the \textit{hedging cost} or arbitrage profit arising from continuous rebalancing. The remaining term defines the LVR:
\begin{equation}
\mathrm{LVR}_t := \frac{1}{2} \int_0^\infty \mathcal{L}_t^q(P) L(q) \, dq. \label{eqn:LVR_local_time1}
\end{equation}
This representation characterizes LVR as a weighted sum of the local times of the price process across the entire price spectrum, with weights provided by the liquidity profile $L$. By the occupation time formula, the LVR can be expressed as a time-integral of the quadratic variation $\langle P \rangle$:
\begin{equation} \label{eqn:LVR_local_time}
\mathrm{LVR}_t = \frac{1}{2} \int_0^t L(P_s) d\langle P \rangle_s,
\end{equation}
This formulation generalizes the LVR results in \cite{milionis2022loss-versus-rebalancing} to arbitrary liquidity profiles and establishes a direct functional link between $L$ and the non-hedgeable risk of the position.

\begin{example}
Consider the edge case $L(q) = l \, \delta(q - q_0)$, where an LP provides liquidity $l$ only at price $q_0$. In this scenario:
\begin{align*}
\mathrm{IL}(P_t | P_0, L) &= l(q_0 - P_t)^+ \1_{\{P_0 \geq q_0\}} + l(P_t - q_0)^+ \1_{\{P_0 \leq q_0\}}, \\
{\rm LVR}_t &= \frac l2 \int_0^\infty \mathcal{L}_t^q(P) \delta(q - q_0) dq = \frac l2 \mathcal{L}_t^{q_0}(P).
\end{align*}
Thus, the IL equates to the payoff of $l$ vanilla options struck at $q_0$, while the LVR is precisely half the local time spent by the price process at $q_0$.
\end{example}

To demonstrate the advantages of expressing LVR via the liquidity profile $L$ as defined in \eqref{eqn:LVR_local_time}, we recover several key results from Sections 6.1 and 6.2 of \cite{fukasawa2023weighted}:

\begin{proposition} \
\begin{enumerate}
\item For a CFMM with the bonding curve $x + \ln y = K$ , the Impermanent Loss at time $T$ is given by: $$\mathrm{IL}_T = -P_T + P_0 + P_T \ln \frac{P_T}{P_0},$$ where $P_T$ and $P_0$ are the spot prices at times $T$ and $0$, respectively. Modulo the P\&L from holding one unit of the risky asset, this IL is equivalent to the payoff of an entropy contract and, at time 0, corresponds to the value of a gamma swap (up to a factor of 2).
\item For a CFMM with the bonding curve $\ln x + y = K$, the Impermanent Loss at time $T$ is: $$\mathrm{IL}_T = \frac{P_T}{P_0} - 1 - \ln \frac{P_T}{P_0}.$$ Modulo the P\&L from holding one unit of the risky asset, this IL is equivalent to the payoff of a log contract and, at time 0, corresponds to the value of a variance swap (up to a factor of 2). 
\end{enumerate}
\end{proposition} 

\begin{proof}
For the curve $x + \ln y = K$, we set $f(x,y) = x + \ln y$. It follows that $p = \frac{f_x}{f_y} = y$. Using \eqref{eqn:ell}, the local intrinsic liquidity is:
\begin{eqnarray*}
\ell = \frac{-2(f_x f_y)^{3/2}}{f_{xx}f_y^2 - 2 f_{xy}f_xf_y + f_{yy}f_x^2} = \frac{-2 y^{-3/2}}{-y^{-2}} = 2\sqrt{y} = 2\sqrt{p}.
\end{eqnarray*}
Thus, the liquidity profile is $L(q) = \frac{\ell(q)}{2\sqrt{q^3}} = \frac{1}{q}$. Applying \eqref{eqn:IL-short} yields:
\begin{eqnarray*}
\mathrm{IL}_T = \int{P_0}^{P_T} (P_T - q)\frac{dq}{q} = P_T \ln\frac{P_T}{P_0} - P_T + P_0.
\end{eqnarray*}
For the curve $\ln x + y = K$, we have $p = \frac{1}{x}$ and $\ell = \frac{2}{\sqrt{p}}$. This results in $L(q) = \frac{1}{q^2}$ , and \eqref{eqn:IL-short} gives:
\begin{eqnarray*}
\mathrm{IL}_T = \int{P_0}^{P_T} (P_T - q)\frac{dq}{q^2} = \frac{P_T}{P_0} - 1 - \ln\frac{P_T}{P_0}.
\end{eqnarray*}
\end{proof}

We remark that the gamma and variance swaps identified above result directly from the LVR. Equation \eqref{eqn:LVR_local_time} implies that for the curve $x + \ln y = K$ (where $L(q) = \frac{1}{q}$), we have:
$$
2 \times \mathrm{LVR}_T = \int_0^T \frac{d\langle P \rangle_s}{P_s},
$$
which is precisely the payoff of a gamma swap. For the curve $\ln x + y = K$, we similarly obtain:
$$
2 \times \mathrm{LVR}_T = \int_0^T \frac{d\langle P \rangle_s}{P_s^2},
$$
representing the payoff of a variance swap.

The converse also holds: if a gamma swap hedges the LVR pathwise such that $\mathrm{LVR}_T = \frac{1}{2} \int_0^T \frac{d\langle P \rangle_t}{P_t}$ , then $L(p) = \frac{1}{p}$. The canonical parametrization then yields:
\begin{eqnarray*}
x(p) = \int_p^C \frac{1}{q} dq = \ln C - \ln p, \quad y(p) = \int_0^p dq = p,
\end{eqnarray*}
for a constant $C > 0$. Eliminating $p$ confirms the bonding curve $x + \ln y = \ln C$. For a variance swap where $\mathrm{LVR}_T = \frac{1}{2} \int_0^T \frac{d\langle P \rangle_t}{P_t^2}$ , the profile $L(p) = \frac{1}{p^2}$ corresponds to the bonding curve $\ln x + y = K$.

\subsubsection{Implications for LVR Design and Pricing}
The formulation of LVR in \eqref{eqn:LVR_local_time} provides a constructive method for designing liquidity profiles, and by extension bonding curves, that satisfy specific risk-management objectives under a time-homogeneous local volatility model. This approach mirrors the ``replicating market maker'' strategy discussed in Section \ref{sec:RMM}, but redirects the focus toward the non-hedgeable LVR component.

Consider a price process $P_t$ governed by the following SDE:
$$
\frac{dP_t}{P_t} = \mu_t dt + \sigma(P_t) dW_t.
$$
The quadratic variation is given by $d\langle P \rangle_t = P_t^2 \sigma^2(P_t) dt$. By selecting a liquidity profile $L(q)$ that inversely scales with the stochastic term of the quadratic variation, specifically:
$$
L(q) = \frac{C}{q^2 \sigma^2(q)},
$$
for some constant $C$, the resulting LVR process becomes:
$$
\mathrm{LVR}_t = \frac{1}{2} \int_0^t L(P_s) d\langle P \rangle_s = \frac{1}{2} \int_0^t \frac{C}{P_s^2 \sigma^2(P_s)} P_s^2 \sigma^2(P_s) ds = \frac{C}{2} t.
$$
In this configuration, the LVR is entirely \textit{deterministic} and proportional to time, effectively eliminating the risk stemming from price volatility. This design methodology allows for the construction of \textit{LVR-neutral profiles} tailored to specific market assumptions.

\begin{example}[LVR-neutral profile in the CEV model]
The Constant Elasticity of Variance (CEV) model is a local volatility model characterized by the SDE:
$$
dP_t = \nu P_t^\beta dW_t
$$ 
for constants $\nu > 0$ and $\beta \geq 0$. Under risk-neutral probability with zero interest and dividend rates, the local volatility function is $\sigma(p) = \nu p^{\beta-1}$. Setting the liquidity profile as 
$$
L(q) = \frac C{q^2\sigma^2(q)} = \frac C{\nu^2 q^{2\beta}},
$$
renders the LVR deterministic at $\frac{C}{2}t$. The bonding curves associated with this profile vary by the elasticity parameter $\beta$:
\begin{eqnarray*}
x &=& \frac C{\nu^2}\int_p^\infty \frac1{q^{2\beta}} dq = \left\{\begin{array}{ll}
\frac C{\nu^2} \frac{p^{1-2\beta}}{2\beta - 1} & \mbox{ if } \beta > \frac12; \\
& \\
\frac C{\nu^2} (\ln M - \ln p) & \mbox{ if } \beta = \frac12, \mbox{ for some large } M > 0; \\
& \\
\frac C{\nu^2}\frac1{1-2\beta} (M^{1-2\beta} - p^{1-2\beta}) & \mbox{ if } \beta < \frac12, \mbox{ for some large } M > 0; 
\end{array}\right. \\
&& \\
y &=& \frac C{\nu^2}\int_0^p \frac1{q^{2\beta - 1}} dq 
= \left\{\begin{array}{ll}
\frac C{\nu^2} \frac1{2(1-\beta)} p^{2(1 - \beta)} & \mbox{ if } \beta < 1; \\
& \\
\frac C{\nu^2} (\ln p - \ln\epsilon) & \mbox{ if } \beta = 1, \mbox{ for some small } \epsilon > 0;; \\
& \\
\frac C{\nu^2}\frac1{2(\beta-1)} (\epsilon^{2-2\beta} - p^{2-2\beta}) & \mbox{ if } \beta > 1, \mbox{ for some small } \epsilon > 0; 
\end{array}\right. 
\end{eqnarray*}
In particular, when $\frac{1}{2} < \beta < 1$, the bonding curve is represented by the G3M form $x^\alpha y^{1-\alpha} = K$, where $\alpha = 2-2\beta$.
\end{example}

Despite its path-dependent nature, LVR can be priced as a static European contingent claim. Let $\Psi(P)$ be the second antiderivative of the liquidity profile, such that $\Psi''(P) = L(P)$. Applying Itô's formula to $\Psi(P_t)$:
\begin{eqnarray}
d\Psi(P_t) &=& \Psi'(P_t)dP_t + \frac{1}{2} \Psi''(P_t) d\langle P \rangle_t \notag \\
&=& \Psi'(P_t)dP_t + \frac{1}{2} L(P_t) d\langle P \rangle_t. \label{eq:ito_psi}
\end{eqnarray}
If the price process $P_t$ is a $\mathbb{Q}$-martingale, integrating from $0$ to $T$ and taking the expectation yields the LVR price:
$$
\mathbb{E}^{\mathbb{Q}}[\mathrm{LVR}_T] = \mathbb{E}^{\mathbb{Q}}[\Psi(P_T)] - \Psi(P_0).
$$
This reduction allows for the use of standard European option pricing techniques, such as Fourier transforms or Monte Carlo simulations, to value the ongoing cost of liquidity provision.

\subsection{Risk Neutral Pricing and Hedging for Impermanent Loss}
In this subsection, we move beyond static replication to address the pricing and hedging of liquidity provision in a dynamic market setting. Building on the equivalence between providing liquidity and maintaining a short position in a convex contingent claim, we apply option pricing theory to quantify risk and establish hedging parameters. By assuming the spot price process $P_t$ is a martingale under a risk-neutral probability $\mathbb{Q}$, we derive the fair value of IL and define implied volatility metrics for general liquidity profiles.  

\subsubsection{Pricing IL as a European-Style Contingent Claim} \label{sec:greeks}
If we treat IL as a European-style contingent claim that pays the amount $\mathrm{IL}(P_T|P_0, L)$ at maturity $T$, its risk-neutral price $\Pi_t^{\mathrm{IL}}(L,T)$ at any time $t \in [0,T]$ is determined by the conditional expectation under the risk-neutral measure $\mathbb{Q}$:
\begin{align}
\Pi^{\mathrm{IL}}_t(L, T) &:= \mathbb{E}^{\mathbb{Q}}\left[ \mathrm{IL}(P_T|P_0,L) \mid \mathcal{F}_t \right] \notag \\
&= \int_0^{P_0} L(q) \mathbb{E}^{\mathbb{Q}}\left[ (q - P_T)^+ \mid \mathcal{F}_t \right] dq + \int{P_0}^\infty L(q) \mathbb{E}^{\mathbb{Q}}\left[ (P_T - q)^+ \mid \mathcal{F}_t \right] dq \notag \\
&= \int_0^{P_0} L(q) P_t(q, T) dq + \int{P_0}^\infty L(q) C_t(q, T) dq, \label{eq:IL_price}
\end{align}
where $C_t(q, T)$ and $P_t(q, T)$ represent the time-$t$ market prices of European call and put options with strike $q$ and maturity $T$.

Applying Jensen’s inequality, it follows that for $t \in [0, T]$:
\begin{equation*}
\Pi_t^{\mathrm{IL}}(L, T) \geq \mathrm{IL}(P_t|P_0, L).
\end{equation*}
This inequality underscores a vital distinction for LPs: the risk-neutral fair value of their potential loss is always greater than or equal to the current realized IL prior to expiry.

\paragraph{\textbf{Delta of IL}}
The Delta of the risk-neutral price $\Pi_t^{\mathrm{IL}}$ with respect to the pool price $P_t$ is given by the weighted integral of the constituent option deltas:
\begin{equation}
\Delta^{\Pi^{\mathrm{IL}}}_t = \frac{\partial \Pi^{\mathrm{IL}}t}{\partial P_t} = \int_0^{P_0} L(q) \Delta^P_t(q,T) dq + \int_{P_0}^\infty L(q) \Delta^C_t(q, T) dq, \label{eq:IL_delta_full}
\end{equation}
where $\Delta^i_t(q, T)$ is the Delta of a put or call option ($i \in \{P, C\}$). This result confirms that within a local volatility framework, the market is complete, and the IL risk can be dynamically hedged by holding $\Delta^{\Pi^{\mathrm{IL}}}_t$ units of the underlying asset $X$.

We distinguish this from the Delta of the current realized IL value, $\Delta_t^{\mathrm{IL}}$, which is defined as:
\begin{eqnarray*}
\Delta_t^{\rm IL} &:=& \frac{\p}{\p P_t}{\rm IL}(P_t|P_0, L) = \frac{\p}{\p P_t}\left\{\int_0^{P_0} L(q) (q - P_t)^+ dq + \int_{P_0}^\infty L(q) (P_t - q)^+ dq \right\} \\
&=& -\int_0^{P_0} L(q) \theta(q - P_t) dq + \int_{P_0}^\infty L(q) \theta(P_t - q) dq \\
&=& \int_{P_0}^{P_t} L(q) dq,
\end{eqnarray*}
where recall that $\theta(\cdot)$ denotes the Heaviside function. Importantly, $\Delta_t^{\mathrm{IL}}$ reflects the instantaneous sensitivity of current reserves but does not generate the terminal payoff required to replicate the IL at maturity in a complete market.

\begin{example}[Delta of Concentrated Liquidity] \label{ex:IL-delta} 
Consider a concentrated liquidity position $L(q) = \frac{\ell}{2q^{3/2}} \1_{[p_a, p_b]}(q)$ for a constant $\ell > 0$ and price boundaries $0 < p_a < P_0 < p_b$. The Delta of the realized Impermanent Loss, $\Delta_t^{\mathrm{IL}}$, represents the instantaneous change in current value:
\begin{align*}
\Delta_t^{\mathrm{IL}} &= \int_{P_0}^{P_t} L(q) , dq = \int_{P_0}^{P_t} \frac{\ell}{2q^{3/2}} \1_{[p_a, p_b]}(q) dq = \begin{cases}
\frac{\ell}{\sqrt{P_0}} - \frac{\ell}{\sqrt{\min\{P_t, p_b\}}} & \text{if } P_t \geq P_0, \\
\frac{\ell}{\sqrt{\max\{P_t, p_a\}}} - \frac{\ell}{\sqrt{P_0}} & \text{if } P_t \leq P_0,
\end{cases}
\end{align*}
whereas the risk-neutral Delta, $\Delta_t^{\Pi^{\mathrm{IL}}}$, which incorporates market expectations and time-to-maturity, is given by:
\begin{align*}
\Delta_t^{\Pi^{\mathrm{IL}}} &= \int_0^{P_0} \Delta_t^P(q, T) L(q) dq + \int_{P_0}^\infty \Delta_t^C(q, T) L(q) dq \notag \\
&= \ell \int_{p_a}^{P_0} \Delta_t^P(q, T) \frac{dq}{2q^{3/2}} + \ell \int_{P_0}^{p_b} \Delta_t^C(q, T) \frac{dq}{2q^{3/2}}.
\end{align*}
\end{example}

\paragraph{\textbf{Gamma of IL}}
The Gamma of the risk-neutral price, $\Gamma_t^{\Pi^{\mathrm{IL}}}$, is the second-order sensitivity to the pool price:
\begin{equation*}
\Gamma^{\Pi^{\mathrm{IL}}}_t = \frac{\partial^2 \Pi_t^{\mathrm{IL}}}{\partial P_t^2} = \int_0^{P_0} L(q) \Gamma^P_t(q, T) dq + \int_{P_0}^\infty L(q) \Gamma^C_t(q, T) dq,
\end{equation*}
where $\Gamma^i_t(q, T)$ is the gamma of the option struck at $q$ and expiry $T$, for $i = \{C, P\}$. 

For the current realized value, the Gamma simplifies directly to the liquidity profile density at the current price:
\begin{equation*}
\Gamma_t^{\mathrm{IL}} = \frac{\partial}{\partial P_t} \Delta_t^{\mathrm{IL}} = L(P_t),
\end{equation*} 

\begin{example}[Gamma of Concentrated Liquidity]
Continuing with the concentrated liquidity profile $L(q) = \frac{\ell}{2q^{3/2}} \1_{[p_a, p_b]}(q)$ from Example \ref{ex:IL-delta}, the Gamma of the realized IL corresponds directly to the liquidity density at the current price:
$$
\Gamma_t^{\mathrm{IL}} = L(P_t) = \frac{\ell}{2\sqrt{P_t^3}} \1_{[p_a, p_b]}(P_t).
$$
The Gamma of the risk-neutral price $\Pi_t^{\mathrm{IL}}$ is the expectation-weighted sum of vanilla option Gammas:
\begin{align*}
\Gamma_t^{\Pi^{\mathrm{IL}}} &= \int_0^{P_0} \Gamma_t^P(q, T) L(q) dq + \int_{P_0}^\infty \Gamma_t^C(q, T) L(q) dq \\
&= \ell \int_{p_a}^{P_0} \Gamma_t^P(q, T) \frac{dq}{2q^{3/2}} + \ell \int_{P_0}^{p_b} \Gamma_t^C(q, T) \frac{dq}{2q^{3/2}}.
\end{align*}
\end{example}

\textbf{Vega of IL}
In models such as Black--Scholes, the Vega of the IL position is likewise a weighted sum of vanilla Vegas:
\begin{equation*}
\nu^{\Pi^{\mathrm{IL}}}_t = \int_0^{P_0} L(q) \nu^P_t(q, T) dq + \int_{P_0}^\infty L(q) \nu^C_t(q, T) dq,
\end{equation*}
where $\nu^i_t(q, T)$ is the vega of the option struck at $q$ and expiry $T$, for $i = \{C, P\}$, enabling LPs to quantify their exposure to forward-looking volatility shifts.

\subsubsection{Implied Volatility for Liquidity Profile} \label{sec:imp-vol}
For a given liquidity profile $L$, we define the \textit{Black--Scholes implied volatility} $\sigma_{\mathrm{BS}}(L, T)$ as the unique value $\sigma$ that satisfies the following equation:
\begin{equation} \label{eqn:implied_vol_BS}
\int_0^{P_0} L(q) P_{\mathrm{BS}}(q, T, \sigma) dq + \int_{P_0}^\infty L(q) C_{\mathrm{BS}}(q, T, \sigma) dq = \int_0^{P_0} L(q) P(q, T) dq + \int_{P_0}^\infty L(q) C(q, T) dq,
\end{equation}
where $P_{\mathrm{BS}}$ and $C_{\mathrm{BS}}$ denote the risk-neutral prices of European puts and calls under the Black--Scholes model. The right-hand side represents the market price of the IL contingent claim, $\Pi^{\mathrm{IL}}$, constructed from discrete market option quotes. Similarly, the \textit{Bachelier implied volatility} $\sigma_{\mathrm{B}}(L, T)$ is defined by solving:
\begin{equation} \label{eqn:implied_vol_B}
\int_0^{P_0} L(q) P_{\mathrm{B}}(q, T, \sigma) dq + \int_{P_0}^\infty L(q) C_{\mathrm{B}}(q, T, \sigma) dq = \int_0^{P_0} L(q) P(q, T) dq + \int_{P_0}^\infty L(q) C(q, T) dq,
\end{equation}
where $P_{\mathrm{B}}$ and $C_{\mathrm{B}}$ are option prices under the Bachelier model. We note that the implied volatilities defined in \eqref{eqn:implied_vol_BS} and \eqref{eqn:implied_vol_B} are homogeneous of degree zero in $L$, as the pricing equations on both sides are linear with respect to the liquidity profile. This metric allows for a standardized ``price of liquidity'' that is consistent with the Black--Scholes paradigm, bridging the gap between endogenous fee generation and exogenous asset volatility.

Since the left-hand side of both \eqref{eqn:implied_vol_BS} and \eqref{eqn:implied_vol_B} is strictly increasing with respect to $\sigma$, we establish the following existence and uniqueness result:

\begin{proposition}[Existence and Uniqueness of Implied Volatility] \label{prop:iv_existence}
For a given liquidity profile $L$ and maturity $T$, there exists a unique solution $\sigma_{\mathrm{BS}}(L,T)$ (resp. $\sigma_{\mathrm{B}}(L,T)$) for equation \eqref{eqn:implied_vol_BS} (resp. equation \eqref{eqn:implied_vol_B}).
\end{proposition}

The central insight is that, for each liquidity profile $L$ and a given maturity $T$, we can assign an implied volatility $\sigma(L, T)$ by equating the theoretical IL price under a specific model (e.g., Black-Scholes or Bachelier) to the option-implied risk neutral price of IL. 

\subsubsection{Fine Structure of Implied Volatility}
While the global implied volatilities $\sigma_{\mathrm{BS}}(L, \cdot)$ and $\sigma_{\mathrm{B}}(L, \cdot)$ characterize the risk profile of the liquidity position in aggregation, they do not distinguish between the risks associated with different segments of the price spectrum. To achieve a more granular assessment, we define the \textit{fine structure of implied volatility}, which highlights the market-implied risk associated with specific intervals within the liquidity profile. This approach is particularly instructive for platforms like Uniswap v3, where liquidity is non-uniformly distributed across discrete tick ranges.

Assume the liquidity profile $L$ is supported on the interval $[p_m, p_M]$ where $0 < p_m < p_M < \infty$. Let $\{\pi_k\}$ be a sequence of nested partitions of this interval such that $\pi_k = \{p_m = p_0 < p_1 < \cdots < p_{i_k} = p_M\}$. For each sub-interval $j \in \{1, \dots, i_k\}$, we define the segmented liquidity profile $L_j^{\pi_k}$ as:
$$
L_j^{\pi_k}(q) = L(q) \1_{\{p_{j-1} \leq q \leq p_j\}},
$$
with its corresponding risk-neutral market price $\Pi_j^{\pi_k}$ defined by the sum of market put and call options integrated over the segment: 
$$
\Pi_j^{\pi_k} := \int_0^{P_0} P(q, T) L_j^{\pi_k}(q) dq + \int_{P_0}^\infty C(q, T) L_j^{\pi_k}(q) dq.
$$
The fine structure of the Black--Scholes implied volatility for the profile $L$ is then defined as the collection of values $\{\sigma_{j, \mathrm{BS}}^{\pi_k}\}$ that solve the pricing equation for each segment $j$:
$$
\Pi_j^{\pi_k, \mathrm{BS}}(\sigma) = \int_0^{P_0} P_{\mathrm{BS}}(q, T, \sigma) L_j^{\pi_k}(q) dq + \int_{P_0}^\infty C_{\mathrm{BS}}(q, T, \sigma) L_j^{\pi_k}(q) dq = \Pi_j^{\pi_k}.
$$
By the linearity of the integral, the total risk-neutral price $\Pi^{\mathrm{IL}}$ is recovered as the sum of the segmented prices:
$$
\Pi^{\mathrm{IL}} = \sum_{j=1}^{i_k} \Pi_j^{\pi_k} = \sum_{j=1}^{i_k} \Pi_j^{\pi_k, \mathrm{BS}}(\sigma_{j, \mathrm{BS}}^{\pi_k}).
$$
However, there are no simple expressions relating the aggregate implied volatility $\sigma_{\mathrm{BS}}(L, \cdot)$ to its fine structure components, the $\sigma_{j, \mathrm{BS}}^{\pi_k}$'s. We conclude that the fine structure for the Bachelier implied volatility can be constructed similarly. We reiterate that examples on the fine structures of implied volatilities for liquidity profiles of various liquidity pools in the Uniswap v3 platform can be found in Section \ref{sec:Numerical}.
\section{Pricing via Last Passage Time}

In this section, we move beyond the standard treatment of Impermanent Loss (IL) as a perpetual Bermudan or American option \cite{bichuch2025price, capponi2025optimal} to propose a valuation framework based on the theory of diffusion processes and last passage time. The core rationale for this approach is the "reset" property of Constant Function Market Makers: \textit{if an LP provides liquidity at a specific price $p$, the associated IL returns to zero whenever the market price revisits that entry level}.

We argue that it is economically rational for an LP to withdraw liquidity when the pool price reaches the entry price $p$ for the final time. Under this rule, the LP's realized IL is minimized (approaching zero) while the duration of liquidity provision is maximized, allowing for optimal compensation through transaction fee collection. Although a last passage time is not a stopping time—making its real-time implementation more subtle than traditional rules—its distribution and associated statistics are well-characterized by the transition densities of diffusion processes.

For the purpose of clear exposition, we assume the pool price $P_t$ follows a geometric Brownian motion (GBM):
$$
\frac{dP_t}{P_t} = \mu dt + \sigma dW_t
$$
where $\mu$ and $\sigma > 0$ are constants and $W_t$ is a standard Brownian motion. The LP is compensated by a transaction fee at a constant rate $\varphi$. Consequently, the LP's discounted P\&L up to a random time $\tau$ is defined as:
\begin{equation} \label{eqn:lpPnL}
-e^{-r\tau} \mathrm{IL}(P_\tau|P_0, L) + \int_0^\tau e^{-rt} \varphi dt = -e^{-r\tau} \mathrm{IL}(P_\tau|P_0, L) + \frac{\varphi}{r}(1 - e^{-r\tau}),
\end{equation}
where the first term represents the discounted loss realized upon exiting the position and the second term represents the accumulated discounted fees.

\subsection{Transiency and the Last Passage Time}
To analyze the behavior of the last passage time, we first characterize the long-term asymptotic properties of the pool price process $P_t$. The scale function $s(p)$ for the geometric Brownian motion defined in the market model is given by:
$$
s(p) = \frac{1}{1-\alpha}\left(p^{1-\alpha} - 1\right),
$$
where $\alpha = \frac{2\mu}{\sigma^2} \neq 1$. The transiency of the process depends on the relationship between the drift $\mu$ and the volatility $\sigma$:
\begin{itemize}
    \item \textit{Downward Transient}: If $\mu < \frac{\sigma^2}{2}$ (i.e., $\alpha < 1$), then $s(0+) = \frac{1}{\alpha - 1} > -\infty$ and $s(\infty) = \infty$. In this regime, the process $P_t$ is transient toward zero, such that $\mathbb{P}\left[ \lim_{t\to\infty} P_t = 0 \mid P_0 = p \right] = 1$ for all $p \in (0, \infty)$.
    \item \textit{Upward Transient}: If $\mu > \frac{\sigma^2}{2}$ (i.e., $\alpha > 1$), then $s(0+) = -\infty$ and $s(\infty) = \frac{1}{\alpha - 1} < \infty$. Here, the process $P_t$ is transient toward infinity, with $\mathbb{P}\left[ \lim_{t\to\infty} P_t = \infty \mid P_0 = p \right] = 1$ for all $p \in (0, \infty)$.
\end{itemize}
For any $\epsilon \in \mathbb{R}$, we define the last passage time $\pi^\epsilon$ at the price level $P_0 e^{\epsilon}$ as:
$$
\pi^\epsilon := \sup\{t \geq 0: P_t = P_0 e^{\epsilon}\} = \sup\{t \geq 0: \nu t + \sigma W_t = \epsilon\},
$$
where $\nu = \mu - \frac{\sigma^2}{2}$. By convention, $\pi^\epsilon = 0$ if the process $P_t$ never visits the price level $P_0 e^{\epsilon}$. In the critical case where $\mu = \frac{\sigma^2}{2}$, the process is recurrent, and the last passage time $\pi^\epsilon = \infty$ almost surely for all $\epsilon \in \mathbb{R}$.

In the subsequent analysis, we focus on the case $\mu > \frac{\sigma^2}{2}$ (upward transient), noting that the results for the downward transient case $\mu < \frac{\sigma^2}{2}$ can be obtained by symmetry.

\subsection{Optimal Withdrawal Strategy}
In the upward transient regime where $\mu > \frac{\sigma^2}{2}$ (and thus $\nu = \mu - \frac{\sigma^2}{2} > 0$), we identify the log-exit price $\epsilon$ that maximizes the liquidity provider's returns. The expected discounted P\&L, $v(\epsilon)$, is defined as the difference between accumulated transaction fees and the realized impermanent loss at the last passage time $\pi^\epsilon$:
\begin{align}
v(\epsilon) &:= \mathbb{E}^{\mathbb{Q}}\left[ -e^{-r\pi^\epsilon} \mathrm{IL}(P_{\pi^\epsilon}|P_0, L) + \frac{\varphi}{r}\left(1 - e^{-r\pi^\epsilon}\right) \mid P_0 \right] \notag \\
&= \frac{\varphi}{r} - \left\{ \int^{P_0 e^{\epsilon}}_{P_0} (P_0 e^{\epsilon} - q) L(q) dq + \frac{\varphi}{r}\right\} \phi_\epsilon(r), \label{eqn:func_v_refined}
\end{align}
where $\phi_\epsilon(r)$ is the moment generating function of the last passage time $\pi^\epsilon$. Following \cite[(36) on p.12]{egami2025decomposition} or \cite[(2.25) on p.28]{profeta2010option}, the function $\phi_\epsilon(r)$ is given by:
$$
\phi_\epsilon(r) = \left\{ \begin{array}{ll}
\displaystyle \frac\nu{\sqrt{\nu^2 + 2r\sigma^2 }} \, e^{\frac\epsilon{\sigma^2}\left(\nu - \sqrt{\nu^2 + 2r\sigma^2}\right)} & \mbox{ if } \epsilon \geq 0; \\
& \\
\displaystyle 1 -  e^{\frac{2\nu}{\sigma^2}\epsilon} + \frac\nu{\sqrt{\nu^2 + 2r\sigma^2 }} \, e^{\frac\epsilon{\sigma^2}\left(\nu + \sqrt{\nu^2 + 2r\sigma^2}\right)} & \mbox{ if } \epsilon < 0.
\end{array}\right.
$$
The value $v(0)$ represents the expected P\&L if the LP withdraws at the final moment the price revisits the initial entry level $P_0$.

\subsubsection{Analytic Characterization of the Optimal Exit}
The determination of the optimal exit level $\epsilon^*$ depends on the sign of the derivative $v'(\epsilon)$:
\begin{eqnarray}
v'(\epsilon) 
&=& \frac{\nu\left(\sqrt{\nu^2 + 2 r \sigma^2} - \nu \right)}{\sigma^2\sqrt{\nu^2 + 2r\sigma^2 }} e^{\frac{\epsilon}{\sigma^2}\left(\nu - \sqrt{\nu^2 + 2r\sigma^2}\right)} \times \nonumber \\
&& \left\{ \frac{\sqrt{\nu^2 + 2 r \sigma^2} - \nu - \sigma^2}{\sqrt{\nu^2 + 2 r \sigma^2} -\nu} \, P_0 e^{\epsilon} \int^{P_0 e^{\epsilon}}_{P_0}  L(q) dq - \int^{P_0 e^{\epsilon}}_{P_0} q L(q) dq +  \frac\varphi r \right\}.  \label{eqn:vprime_sign}
\end{eqnarray}
Note that $v'(0^+) = \lim_{\epsilon \to 0^+} v'(\epsilon) = - \frac{\nu}{\sigma^2\sqrt{\nu^2 + 2r\sigma^2}} \frac{\varphi}{r} (\nu - \sqrt{\nu^2 + 2r\sigma^2}) > 0$, indicating that $v$ is increasing to the right of the origin.

Utilizing the identity $\nu^2 + 2r\sigma^2 - (\sigma^2 + \nu)^2 = 2(r - \mu)\sigma^2$, the behavior of $v(\epsilon)$ is categorized by the relationship between the drift $\mu$ and the discount rate $r$:
\begin{enumerate}
    \item \textit{Asymptotic Maximization ($\mu < r$):} If the bracketed term in \eqref{eqn:vprime_sign} remains positive for all $\epsilon > 0$, $v$ is monotonically increasing. This typically occurs when $\mu < r$, as the first term in the brackets is positive and increasing. In this regime, the LP optimally maintains the position indefinitely, with the expected P\&L approaching the supremum $\frac{\varphi}{r}$ as $\epsilon \to \infty$.
    \item \textit{Unique Interior Maximizer ($\mu \geq r$):} If a unique $\epsilon^* > 0$ exists such that $v'(\epsilon^*) = 0$, then $v$ is maximized at this point. This uniqueness is guaranteed for $\mu \geq r$ because the first term in the brackets decreases toward $-\infty$ as $\epsilon \to \infty$.
\end{enumerate}
When $\mu = r$, the identity $\nu + \sigma^2 - \sqrt{\nu^2 + 2 r \sigma^2} = 0$ holds. Consequently, the first-order condition $v'(\epsilon) = 0$ simplifies to $\int_{P_0}^{P_0 e^{\epsilon}} q L(q) \, dq = \frac{\varphi}{r}$. This yields the explicit solution:
$$
P_0 e^{\epsilon^*} = y^{-1}\left(\frac{\varphi}{r} - y(P_0) \right),
$$
where $y(\cdot)$ is the reserve function defined in \eqref{eq:reserves}. These results are summarized in Theorem \ref{thm:optimal_exit}.

\begin{theorem} \label{thm:optimal_exit}
Let the pool price $P_t$ follow the geometric Brownian motion $\frac{dP_t}{P_t} = \mu dt + \sigma dW_t$, with $\sigma > 0$ and $\mu > \frac{\sigma^2}{2}$. The expected discounted P\&L $v(\epsilon)$, stopped at the last passage time $\pi^\epsilon$ of the price level $P_0 e^\epsilon$, is given by:
$$
v(\epsilon) = \mathbb{E}^{\mathbb{Q}} \left[ -e^{-r\pi^\epsilon} \mathrm{IL}(P_{\pi^\epsilon}|P_0, L) + \frac{\varphi}{r}(1 - e^{-r\pi^\epsilon}) \mid P_0 \right].
$$
The optimal exit level $\epsilon^*$ is characterized as follows:
\begin{enumerate}
    \item If $\mu \ge r$, there exists a unique maximizer $\epsilon^* > 0$ satisfying the transcendental equation:
    \begin{equation} \label{eqn:optimal_epsilon}
    \frac{\sqrt{\nu^2 + 2r\sigma^2} - \nu - \sigma^2}{\sqrt{\nu^2 + 2r\sigma^2} - \nu} P_0 e^\epsilon \int_{P_0}^{P_0 e^\epsilon} L(q) dq - \int_{P_0}^{P_0 e^\epsilon} q L(q) dq + \frac{\varphi}{r} = 0,
    \end{equation} 
    where $\nu = \mu - \frac{\sigma^2}{2}$.
    \item In the specific case $\mu = r$, the optimal exit price admits the explicit solution: $$P_0 e^{\epsilon^*} = y^{-1} \left( \frac{\varphi}{r} - y(P_0) \right), $$ where $y$ is the reserve function defined in \eqref{eq:reserves}.
    \item If $\mu < r$, $v(\epsilon)$ may be monotonically increasing for all $\epsilon \in \mathbb{R}$, in which case \eqref{eqn:optimal_epsilon} has no solution and $v$ attains its supremum $\frac{\varphi}{r}$ as $\epsilon \to \infty$.
\end{enumerate}
\end{theorem}

\begin{example} \label{ex:withdrawal_cases}
Consider an initial pool price $P_0 = 1$ and a liquidity profile $L(q) = e^{-\frac{1}{10}|q-1|}$. With volatility $\sigma = 10\%$ and fee rate $\varphi = 2\%$, we evaluate the optimal exit strategy under varying drift $\mu$ and discount rate $r$:
\begin{itemize}
    \item \textit{Case 1 ($\mu=2\%, r=1\%$):} Since $\mu > r$, Equation \eqref{eqn:optimal_epsilon} yields a unique finite maximizer $\epsilon^* \approx 0.65$, as seen in the left panel of Figure \ref{fig:func_v}.
    \item \textit{Case 2 ($\mu=2\%, r=3\%$):} Despite $\mu < r$, the specific interplay of parameters allows for a unique interior maximizer $\epsilon^* \approx 0.45$, illustrated in the middle panel of Figure \ref{fig:func_v}.
    \item \textit{Case 3 ($\mu=1\%, r=4\%$):} Here, the drift is significantly lower than the discount rate. The equation has no solution, and $v(\epsilon)$ increases toward its supremum $\frac{\varphi}{r} = 0.5$ as $\epsilon \to \infty$, shown in the right panel of Figure \ref{fig:func_v}.
\end{itemize}

\begin{center}
    \begin{minipage}[b]{0.32\textwidth}
        \centering
        \includegraphics[width=\textwidth]{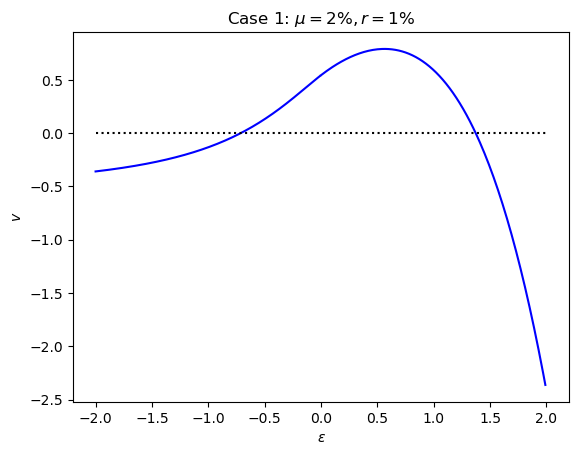}
    \end{minipage}
    \hfill
    \begin{minipage}[b]{0.32\textwidth}
        \centering
        \includegraphics[width=\textwidth]{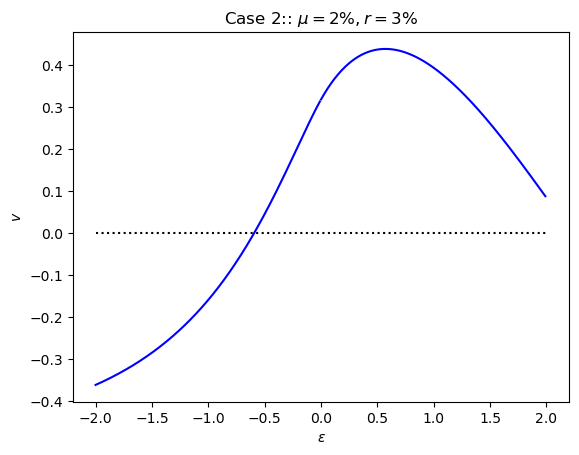}
    \end{minipage}
    \hfill
    \begin{minipage}[b]{0.32\textwidth}
        \centering
        \includegraphics[width=\textwidth]{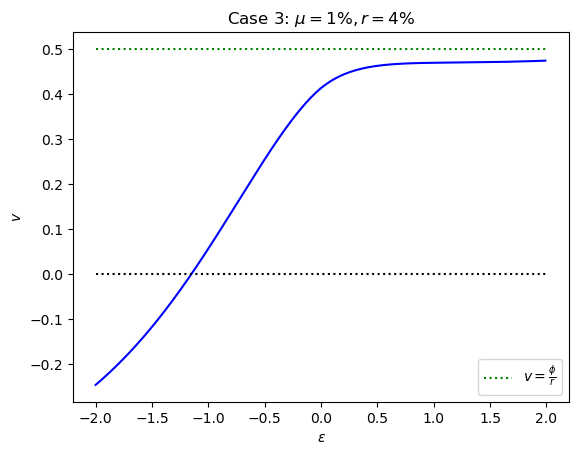}
    \end{minipage}
    \captionof{figure}{The expected P\&L $v(\epsilon)$ defined in \eqref{eqn:optimal_epsilon} as a function of the log-exit price $\epsilon$. \textbf{Left (Case 1):} A unique maximizer $\epsilon^* > 0$ exists because $\mu > r$. \textbf{Middle (Case 2):} A unique interior maximizer exists despite $\mu < r$, due to the interplay of fee income and last-passage statistics. \textbf{Right (Case 3):} $v(\epsilon)$ increases monotonically toward the supremum $\frac{\varphi}{r} = 0.5$ as $\epsilon \to \infty$, indicating an optimal strategy of never withdrawing.}
    \label{fig:func_v}
\end{center}
\end{example}

\section{Data and Numerical Implementation} \label{sec:Numerical}

In this section, we present empirical results on the fine structure of implied volatility for Uniswap V3 liquidity profiles, utilizing the framework developed in Section \ref{sec:Pricing}. The central advancement of this framework is the impermanent loss (IL) replication formula \eqref{eqn:IL_replication}, which characterizes IL as a weighted strip of vanilla options. For a liquidity profile $L(q) = \ell(q)/(2q^{3/2})$, we can assign an implied volatility to any liquidity provision by equating the option-implied IL price to its corresponding model price under Black-Scholes or Bachelier dynamics.

A primary advantage of this canonical parametrization is that it enables a direct economic comparison across disparate liquidity pools. Specifically, pools with varying fee tiers or liquidity concentrations can be evaluated on a unified basis of implied volatility. In the following analysis, we compare Black-Scholes and Bachelier implied volatilities across multiple expiries and fee tiers using market data.

\subsection{Data Sources and Collection}
This subsection details the empirical data acquisition process required to synthesize the market-side IL integral and reconstruct the corresponding liquidity profiles.

\subsubsection{Option Data}
European option quotes for ETH are obtained from the Deribit exchange through its public API. For each hourly snapshot, the following data points are collected:
\begin{itemize}
    \item \textit{Bid-ask mid prices}: Prices for both call and put options are gathered across the full spectrum of available strikes.
    \item \textit{Quarterly expiries}: The dataset includes four specific expiries: December 2025, March 2026, June 2026, and September 2026.
    \item \textit{Forward Price and Maturity}: We collect the underlying reference forward price $F$ (Deribit \texttt{underlying\_price}) and the time to maturity $T$ in years for each respective expiry.
\end{itemize}
Let $P_0$ denote the Uniswap pool spot price at the snapshot, and let $F$ denote the Deribit maturity-$T$ reference forward used for the option surface. In our pipeline, $P_0$ is used only to implement the IL replication split in \eqref{eq:IL_price} (i.e., the boundary between the put and call legs), while $F$ is used for put--call parity, model option prices, and the log-moneyness coordinate $x=\log(K/F)$. Consequently, the IL split point $K=P_0$ corresponds to $x=\log(P_0/F)$ on the plots.

\subsubsection{Liquidity Data}
On-chain liquidity profiles are extracted directly from the Uniswap V3 ETH/USDC pool contracts. For each snapshot, the following parameters are queried:
\begin{itemize}
    \item \textit{Tick Data}: All initialized ticks and their corresponding signed liquidity deltas $\Delta \ell_i$.
    \item \textit{Pool State}: The pool's current tick, reported in-range liquidity $\ell_{\mathrm{curr}}$, and the pool price $P_0$.
\end{itemize}
The intrinsic liquidity $\ell(q)$—representing active in-range liquidity at price $q$—is reconstructed by aggregating the signed tick deltas. Here, $\Delta \ell_i$ corresponds to the pool's \texttt{liquidityNet}, representing the change in active liquidity as the price crosses tick $t_i$ upward \cite{uniswap2023math}. To fix the additive constant of the cumulative sum, the profile is anchored to the reported in-range liquidity $\ell_{\mathrm{curr}}$ at the current tick $t_0$, setting $\ell(P_0^-) = \ell_{\mathrm{curr}}$.

The analysis focuses on two ETH/USDC pools with fee tiers of 5 basis points (5bp) and 30 basis points (30bp). These pools allow for a contrast between tight liquidity concentration (5bp) and broader spatial distributions (30bp). The main-text empirical results use the November 17, 2025 snapshot, while additional weekly snapshots are reported in Appendix \ref{sec:appendix-snapshots}.

\subsection{Data Cleaning and Processing}

\subsubsection{Option Data Cleaning}
To ensure the integrity of the market-side IL integrals, raw option quotes from Deribit are subjected to standard no-arbitrage filters as recommended by CBOE \cite{cboe2019options}. The following constraints are applied:
\begin{itemize}
    \item \textit{Positive Prices}: All quotes with non-positive mid prices are removed.
    \item \textit{Monotonicity}: Call prices must be non-increasing with respect to the strike, while put prices must be non-decreasing.
    \item \textit{Convexity}: To prevent butterfly arbitrage, the second difference in prices relative to the strike must be non-negative (i.e., $C(K_{i-1}) - 2C(K_i) + C(K_{i+1}) \geq 0$).
\end{itemize}
In our analyzed snapshots, OTM quotes passed these filters after positive-price filtering. ITM coverage gaps, where arbitrage violations are more common, are addressed implicitly via put–call parity.

\subsubsection{Interpolation and Synthetic Pricing}
The integration variable is identified with the strike $K$ (measured in USD per ETH), which aligns with the price axis $q$ used in the liquidity profile. Between observed strikes, option prices are linearly interpolated. For strikes $K \in [K_i, K_{i+1}]$ with observed prices $C_i$ and $C_{i+1}$, the call price is given by:
$$
C(K) = C_i + \frac{C_{i+1} - C_i}{K_{i+1} - K_i}(K - K_i).
$$
This piecewise linear representation, when combined with the piecewise constant liquidity profile derived from the Uniswap tick structure, allows for the use of closed-form antiderivatives for the IL integrals. Specifically, on any interval $[a, b]$ where $\ell$ is constant, thus $L(q) = \ell / (2q^{3/2})$, and the option price $O(q) = a_0 + a_1 q$ is linear:
$$
\int_a^b L(q) O(q) , dq = \frac{\ell}{2} \left[ -\frac{2a_0}{\sqrt{q}} + 2a_1 \sqrt{q} \right]_a^b.
$$
This analytic approach completely sidesteps numerical quadrature noise. The piecewise-linear interpolant is the pointwise \textit{maximal} arbitrage-free fill-in between observed strikes \cite{cohen2020detecting}. Since $L(q)$ is positive, this upper bound on the market price translates into a conservative benchmark for the subsequent implied-volatility inversion.

Higher-order spline interpolation is avoided for two reasons \cite{lefloch2020arbitrage}:
\begin{itemize}
    \item \textit{Static Arbitrage}: Splines are not necessarily shape-preserving and can violate the convexity constraints required to exclude butterfly arbitrage \cite{wystup2017arbitrage}.
    \item \textit{Numerical Ringing}: Cubic splines may overshoot between knots, and these artifacts can be amplified into spurious high-frequency oscillations in the implied volatility curve during the nonlinear inversion process.
\end{itemize}

\paragraph{\textit{Synthetic option pricing.}}
When strike coverage on one side of the market is sparse, missing prices are synthesized from the better-covered side using put--call parity. With $r=0$, this is $C(K)-P(K)=F-K$ (equivalently, $P=C-F+K$ and $C=P+F-K$). This approach guarantees continuous coverage across the integration domain while preserving no-arbitrage conditions, as put--call parity is an exact identity.

The threshold for switching to synthetic pricing is a local strike gap of 500 USD or more. This data-coverage heuristic reflects the transition from dense near-ATM strike spacing on Deribit (typically 25–50 USD) to sparse deep-OTM spacing (500 USD or more). We found this to be a stable cutoff across expiries in our snapshot. When gaps exceed this threshold, linear interpolation becomes unreliable, and we instead rely on the synthetic prices derived from the more liquid side of the market.

\paragraph{\textit{Market Proxy vs. Model Prices.}}
Our empirical pipeline utilizes two distinct objects to facilitate the implied-volatility inversion:
\begin{itemize}
    \item \textit{Market Proxy} ($\widehat O^{\rm mkt}(K)$): This is constructed from discrete market option quotes via no-arbitrage cleaning and piecewise-linear interpolation. It is used exclusively to evaluate the market-side IL integrals. We maintain this evaluation as analytically as possible by using closed-form antiderivatives to avoid injecting numerical quadrature noise into the subsequent inversion process.
    \item \textit{Model Option Prices} ($O^{\rm BS}(K;\sigma)$ and $O^{\rm B}(K;\sigma)$): These are the theoretical prices used to define Black–Scholes and Bachelier implied volatilities. Inversion is performed by matching model-side and market-side IL integrals on each integration bin at the chosen resolution.
\end{itemize}
By partitioning the integration domain at all Uniswap V3 tick boundaries, strike knots, and the split point $P_0$, we ensure that the intrinsic liquidity $\ell$ remains constant and the market proxy remains affine on each segment. This structure allows for an exact closed-form evaluation of the market-side integral, providing a robust and noise-free benchmark for determining the fine structure of implied volatility.

\subsubsection{Integration Partitioning}
To maintain analytic precision, the integration domain is partitioned at every point where the integrand's functional form changes. These boundaries include all Uniswap V3 tick limits where the intrinsic liquidity $\ell$ shifts, all strike knots where the market proxy $\widehat{O}^{\text{mkt}}(K)$ undergoes a change in slope, and the spot price $P_0$, which serves as the put–call split point. Because liquidity $\ell$ remains constant and the market proxy remains affine on each resulting segment, the market-side integral admits an exact closed-form evaluation.

\subsubsection{Numerical Integration Details}
The IL integrals, characterized by the form 
$$
\int_a^b \frac{O(K)}{K^{3/2}} \, dK,
$$
are evaluated using strategies tailored to the specific model dynamics. For market prices and the Black–Scholes model ($r=0$), we utilize piecewise linear interpolation and closed-form antiderivatives for
$$
\int_a^b \frac{C_{\mathrm{BS}}(K;\sigma)}{K^{3/2}} \, dK
$$
to avoid injecting numerical quadrature noise into the evaluation. In the Bachelier model, where the auxiliary term involving $\Phi(d(K))/\sqrt{K}$ lacks an elementary antiderivative, we employ a hybrid approach: exact boundary terms are combined with 32-point Gauss–Legendre quadrature on the smooth remainder. This semi-closed-form method ensures machine-precision accuracy across each segment; see Appendix~\ref{sec:bachelier-quadrature} for the derivation and numerical validation.

\subsubsection{Interest Rate Convention}
In accordance with Deribit exchange conventions, the risk-free rate is set to $r = 0$ throughout the implementation. This is consistent with the theoretical assumption of equal interest rates for both tokens in the pool in Section \ref{sec:Pricing}. While the empirical forward price $F$ is used for model pricing and put–call parity identities, the pool price $P_0$ is used exclusively as the boundary to split the integral between put and call terms in \eqref{eq:IL_price}.

\subsection{Implied Volatility Computation (Multi-Resolution)}
Let $x := \log(K/F)$ denote log-moneyness (identifying the integration variable $q$ with strike $K$). For a given resolution $n$, we define a partition
$x_0 < x_1 < \cdots < x_n$ and corresponding bins
$B_{n,k} := \{ K : x(K) \in [x_{k-1}, x_k)\}$, $k=1,\dots,n$.
In practice, the partition is taken as a coarsening of the finest initialized-tick partition, obtained by aggregating adjacent tick buckets into $n$ bins.
We then restrict the liquidity profile $L$ to each bin via
\[
L_{n,k}(K) := L(K)\1_{\{K\in B_{n,k}\}}.
\]

For each bin we compute Black--Scholes and Bachelier implied volatilities
$\sigma^{\rm BS}_{n,k}$ and $\sigma^{\rm B}_{n,k}$ by matching model and market IL prices:
\[
\int L_{n,k}(K)\,\widehat O^{\rm mkt}(K)\,dK
=
\int L_{n,k}(K)\,O^{\rm BS}(K;\sigma^{\rm BS}_{n,k})\,dK,
\]
and the same for Bachelier.  By Proposition \ref{prop:iv_existence}, each equation admits a unique solution because the combined put+call map is strictly increasing in $\sigma$. We solve by bisection and do not invert put and call components separately.

We report multiple resolutions $n \in \{1,3,6,12,N\}$, where $N$ is the finest initialized-tick partition. The Black--Scholes IV is reported as an annualized volatility level, and the Bachelier IV is normalized as
\[
\bar\sigma_{\rm B}:=\sigma_{\rm B}/P_0,
\]
matching the scale used in the multi-resolution overlay figures.

\subsection{Empirical Results: Multi-Resolution Implied Volatility}
The main-text empirical results use the November 17, 2025 snapshot for both 5bp and 30bp ETH/USDC pools. Additional weekly snapshots are reported in the Appendix \ref{sec:appendix-snapshots} as robustness checks.

\subsubsection{Liquidity Concentration and IL Contribution}
Figures \ref{fig:il_30bps} and \ref{fig:il_5bps} report IL contribution for a representative expiry (March 2026) in each fee tier. In this snapshot, the 5bp pool concentrates weight more tightly near $P_0$, while the 30bp pool distributes weight more broadly.

\begin{center}
    \includegraphics[width=0.95\textwidth]{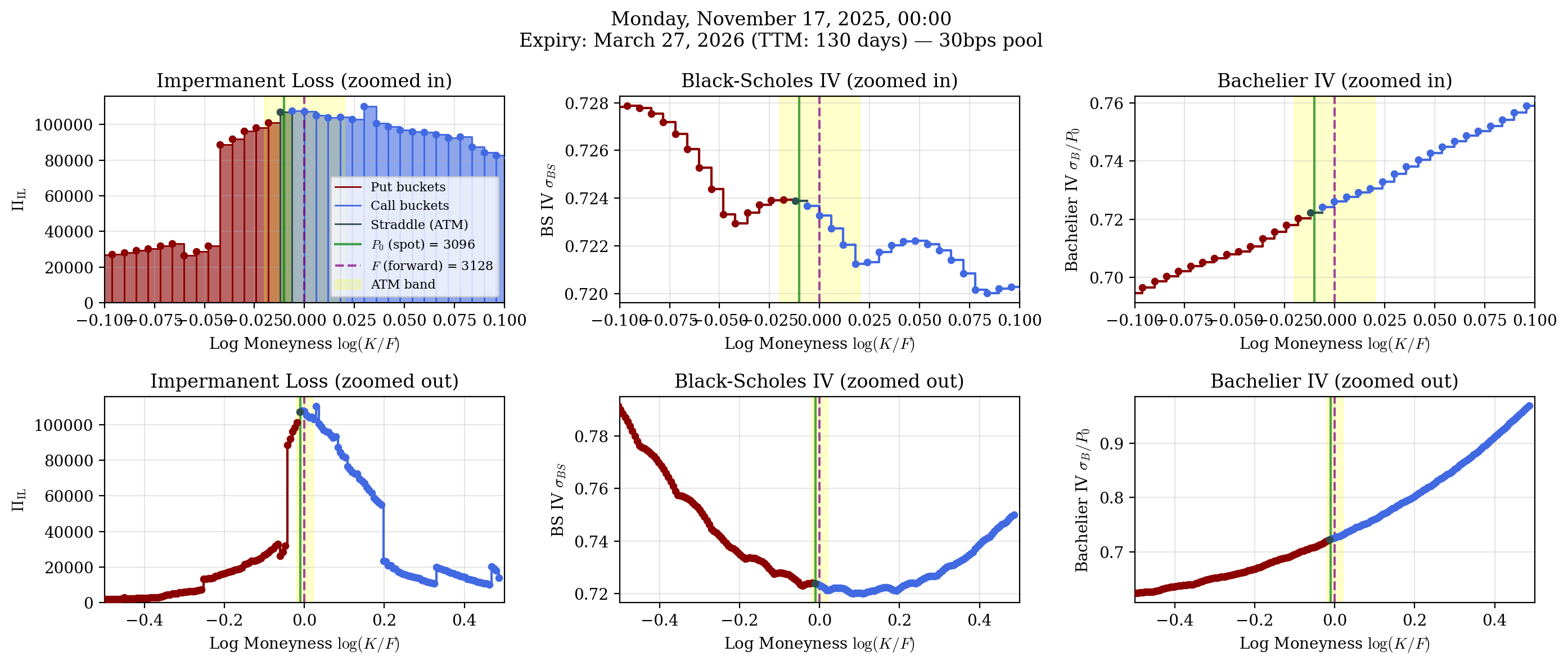}
    \captionof{figure}{IL contribution for the 30bp ETH/USDC pool (Nov 17, 2025 snapshot; representative expiry). The left column reports IL contribution; the remaining columns show the corresponding single-resolution IV outputs from the shared plotting routine.}
    \label{fig:il_30bps}
\end{center}

\begin{center}
    \includegraphics[width=0.95\textwidth]{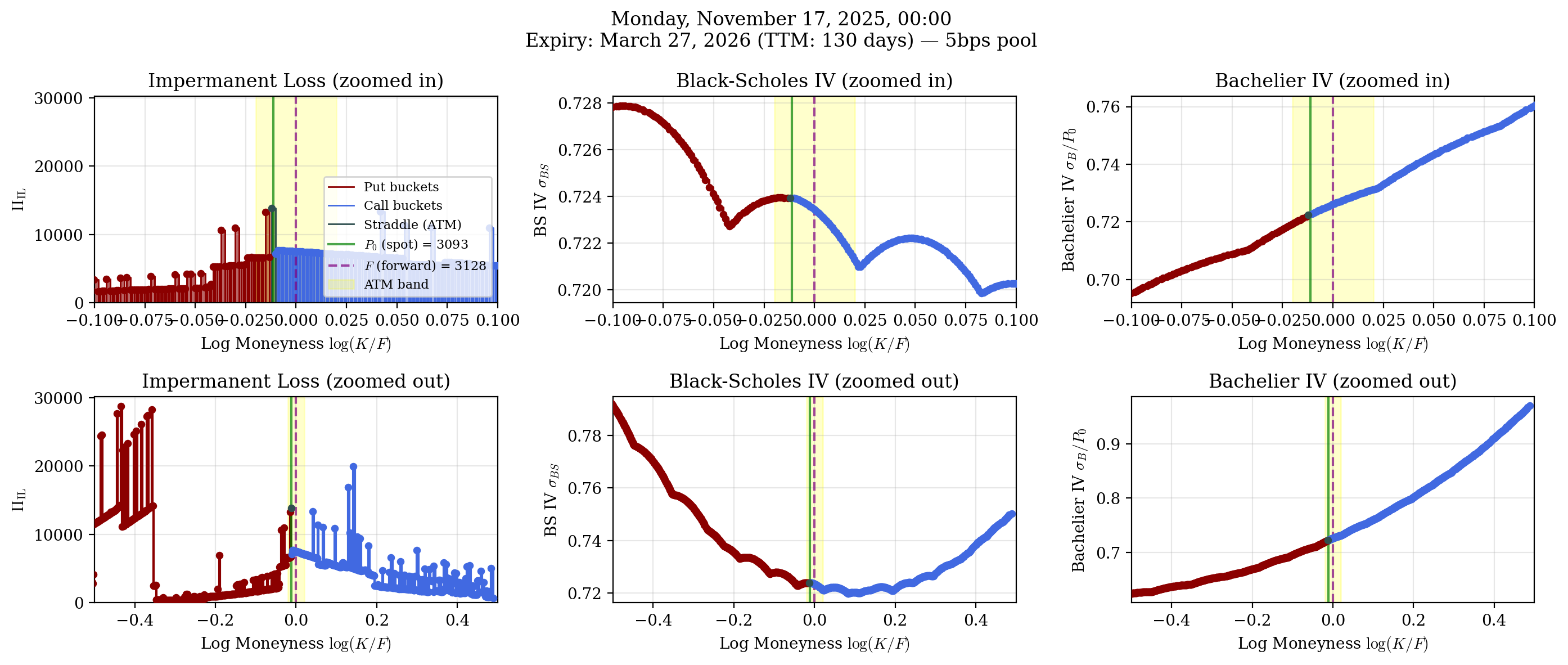}
    \captionof{figure}{IL contribution for the 5bp ETH/USDC pool (Nov 17, 2025 snapshot; representative expiry). The left column reports IL contribution; the remaining columns show the corresponding single-resolution IV outputs from the shared plotting routine.}
    \label{fig:il_5bps}
\end{center}

Having established how each fee tier loads the IL strip, we now turn to the implied-volatility surfaces obtained by equating IL prices across models.

Figures \ref{fig:multires_30bps} and \ref{fig:multires_5bps} overlay implied-volatility curves across resolutions $n\in\{1,3,6,12,N\}$ for four expiries (columns). The top row shows Black--Scholes IV and the bottom row shows normalized Bachelier IV $\bar\sigma_{\rm B}=\sigma_{\rm B}/P_0$.

\begin{center}
    \includegraphics[width=0.98\textwidth]{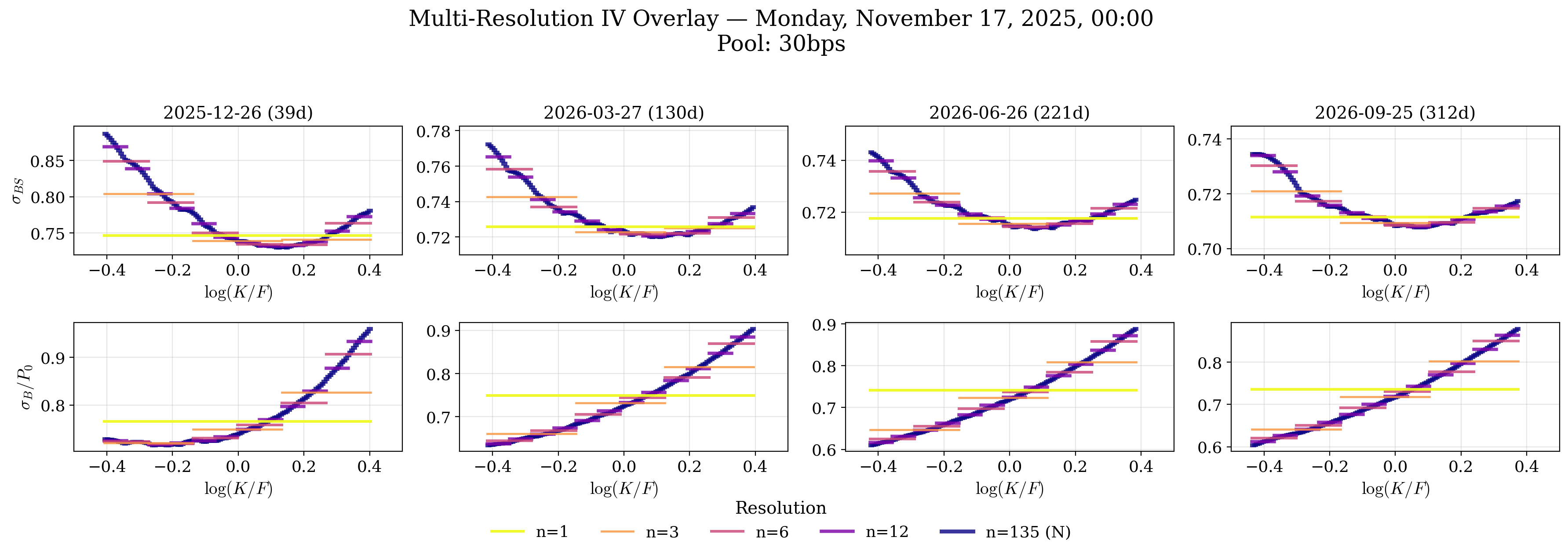}
    \captionof{figure}{Multi-resolution implied-volatility overlay for the 30bp ETH/USDC pool (Nov 17, 2025 snapshot). Curves correspond to $n\in\{1,3,6,12,N\}$, with $N$ the finest initialized-tick partition. Columns are quarterly expiries; rows are Black--Scholes IV (top) and normalized Bachelier IV $\bar\sigma_{\rm B}=\sigma_{\rm B}/P_0$ (bottom).}
    \label{fig:multires_30bps}
\end{center}

\begin{center}
    \includegraphics[width=0.98\textwidth]{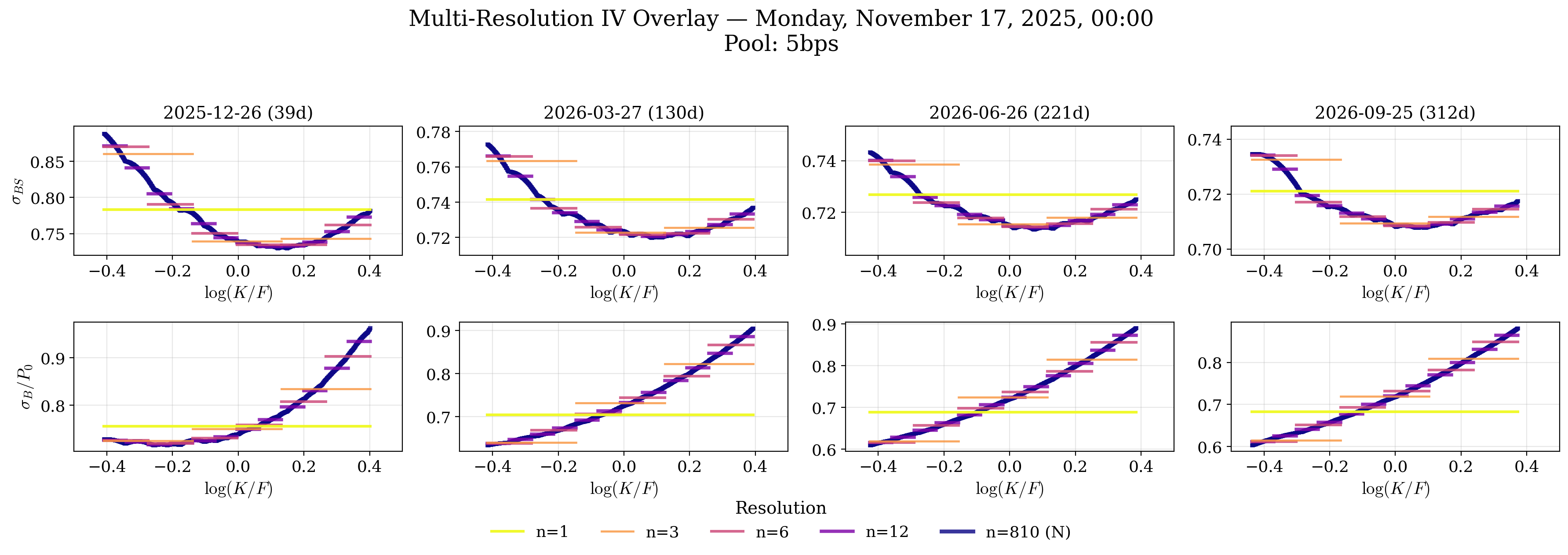}
    \captionof{figure}{Multi-resolution implied-volatility overlay for the 5bp ETH/USDC pool (Nov 17, 2025 snapshot). Fine-scale oscillations are most visible at $n=N$, while coarser aggregations (e.g., $n=6$ and $n=12$) recover a stable macro smile/skew structure.}
    \label{fig:multires_5bps}
\end{center}

The overlays support four robust empirical findings:
\begin{itemize}
    \item \textit{Scale separation}: the macro smile and skew are already visible at coarse resolutions and remain qualitatively stable as $n$ increases.
    \item \textit{Microstructure concentration}: high-frequency wiggles are concentrated at the finest resolution, consistent with small-bin conditioning effects rather than large economic regime shifts.
    \item \textit{Term structure}: near-dated expiries show stronger curvature, while longer-dated expiries are flatter in both Black--Scholes and Bachelier representations.
    \item \textit{Fee-tier comparison}: the 5bp pool exhibits finer native structure at $n=N$, but at matched coarse resolution (e.g., $n=12$) the 5bp and 30bp smiles are closely aligned. This coincides with tighter IL concentration in Figures \ref{fig:il_30bps}--\ref{fig:il_5bps}, which can make small-bin inversions more sensitive.
\end{itemize}

\subsection{Summary Statistics}
Under the multi-resolution framing, extreme min/max values at the native finest partition are sensitive to microstructure and are not the preferred primary summary for economic interpretation. We therefore emphasize the cross-resolution overlays in Figures \ref{fig:multires_30bps} and \ref{fig:multires_5bps}; additional weekly snapshots are provided in Appendix~\ref{sec:appendix-snapshots}.

\subsection{Key Takeaways}
The empirical analysis of multi-resolution implied volatility yields several significant conclusions regarding the pricing of liquidity provision in Uniswap V3:
\begin{enumerate}
    \item \textit{Persistence of the Volatility Smile}: The IL-implied volatility exhibits the characteristic "smile" shape prevalent in traditional options markets, with elevated IV for out-of-the-money strikes on both sides. This confirms that the option-based replication formula \eqref{eqn:IL_replication} effectively inherits the market's complex view on price dynamics and tail risk.
    \item \textit{Model Convergence near ATM}: Near the money, Black--Scholes IV and normalized Bachelier IV are closely aligned; divergence is primarily a wing effect.
    \item \textit{Term Structure Dynamics}: The overlays reveal systematic maturity smoothing: near-term expiries exhibit stronger smile curvature, while longer-term expiries are flatter.
    \item \textit{Decoupling of Fee Tiers and Economic Risk}: While the 5bp and 30bp pools have very different native bucket granularity, their coarse-resolution smiles remain close. This indicates fee tiers primarily reshape liquidity placement rather than the market-implied level of economic risk.
    \item \textit{Unified Cross-Pool Comparison}: The canonical parametrization using $(P, \ell)$ enables direct comparison across fee tiers and resolutions. The persistence of similar coarse-scale IV levels across pools suggests LP risk is largely a market-level property.
    \item \textit{Level}: The inferred Black--Scholes IV levels are consistent with the Deribit option surface used as input in our snapshots.
\end{enumerate}

\section{Conclusion and Discussion}
We have demonstrated in this article that the canonical parametrization of the bonding curve in a CFMM, as established in \eqref{eq:reserves}, provides a robust and convenient framework for addressing several key problems in decentralized finance. Specifically, this formulation enables:
\begin{itemize}
    \item \textit{Risk-Neutral Pricing and Hedging}: The derivation of fair values and sensitivities (Delta, Gamma, Vega) for both pool value and impermanent loss.
    \item \textit{Implied Volatility Fine Structure}: The assignment of model-implied volatilities to discrete liquidity segments, allowing for a granular assessment of risk across the price spectrum.
    \item \textit{Loss-Versus-Rebalancing (LVR)}: A generalized understanding of the non-hedgeable costs of liquidity provision as a time-integral of quadratic variation weighted by the liquidity profile.
    \item \textit{Path-Dependent Analysis}: The pricing of impermanent loss using the last-passage time theory, which accounts for the "reset" property of CFMM positions.
\end{itemize}

Numerical examples utilizing market data from Uniswap V3 ETH/USDC pools and Deribit option markets validate these methodologies. Our empirical results confirm a volatility smile consistent with crypto-asset dynamics and demonstrate that while fee tiers alter the structural concentration of liquidity, the underlying economic risk remains consistent across the market.

Throughout this work, except for the numerical analysis, liquidity profiles have been treated as static. In practice, while the intrinsic liquidity profile $\ell$ is observable on-chain, its evolution is dynamic, driven by block-level validation and LP activity. This observation motivates the shift toward statistical and dynamical modeling of liquidity surfaces. Preliminary empirical analyses of the temporal evolution of these profiles can be found in \cite{risk2025dynamics}.

Furthermore, the parametrization in \eqref{eq:reserves} provides a foundation for developing market models that are fully consistent with the CFMM mechanism. By accounting for the dual nature of pool updates—where trades move the state along the bonding curve while LP additions or withdrawals shift the curve itself—the evolution of the pool state $(x_t, y_t, P_t)$ can be modeled as a dynamical system:
\begin{align*} 
dx_t &= \int_{P_t}^\infty dL_t(q) dq + d\Lambda^a_t - d\Lambda^b_t, \notag\\
dy_t &= \int_0^{P_t} q dL_t(q) dq - P_t \left( d\Lambda^a_t - d\Lambda^b_t \right), \\
dP_t &= -\frac1{L_t(P_t)} \left( d\Lambda^a_t - d\Lambda^b_t \right),\notag \\
dL_t(q) &= \left\{ \mathcal{A} \, L_t(q) + h(t, q) \right\} dt,
\end{align*}
where $\Lambda^i_t$ represents cumulative trade volumes and the operator $\mathcal{A}$ describes the diffusive properties of liquidity over time. We refer the reader to \cite{lee2026dynamics} for a stochastic treatment of the dynamic for $L$. These equations capture the complex interplay between order flow, liquidity provision, and price formation. The comprehensive analysis of this system, particularly under stochastic noise, remains an area for future research.

\section*{Acknowledgements}
S. N. T. is grateful for the financial support from the National Science and Technology Council of Taiwan under grant 114-2115-M-007-012-MY3, "Mathematical Foundation of Automated Market Makers."

The authors used Google Gemini 3 (Flash version, March 2, 2026) to polish the manuscript's grammar and style. All original theorems, proofs, and analyses were developed solely by the authors, who assume full responsibility for the integrity of the work.
\appendix

\section{Note on Interest Rates} \label{sec:interest-rates} 
This section details the interest rate conventions and forward price mechanics utilized in the pricing models, specifically relating them to standard Black-Scholes dynamics. We define the relevant rates as follows:
\begin{itemize}
    \item $r$: The risk-free rate of the numeraire (e.g., USD).
    \item $\delta$: The yield or staking reward of the underlying asset (e.g., ETH). This is functionally equivalent to a foreign short rate in foreign exchange (FX) market models.
\end{itemize}
The risk-neutral drift for the underlying asset under the numeraire measure is determined by the differential $r - \delta$. Consequently, the forward price $F_{0,T}$ at time $T$ is defined by:
$$
F_{0,T} = P_0 e^{(r-\delta)T}.
$$
Under these definitions, the standard put-call parity identity holds, mirroring the convention used in FX markets:
$$
C - P = e^{-rT}(F_{0,T} - K) = P_0 e^{-\delta T} - K e^{-rT}.
$$

In the numerical implementation discussed in Section \ref{sec:Numerical}, the risk-free rate is set to $r = 0$ throughout, consistent with the assumption of equal interest rates for both tokens in the pool. Under this $r = \delta = 0$ assumption, the theoretical forward price $F$ equals the spot price $S_0$. Empirically, however, the Deribit forward price (\texttt{underlying\_price}) may differ slightly from the Uniswap pool price $P_0$; we utilize the forward $F$ for all pricing and parity calculations while maintaining $P_0$ as the put/call split point.

\section{Numerical Integration for Bachelier IL} \label{sec:bachelier-quadrature}
Both Black-Scholes and Bachelier models require the evaluation of integrals in the form:
$$
\int_a^b \frac{O(K)}{K^{3/2}} \, dK,
$$
where $O(K)$ represents an option price. Applying integration by parts with $f = O(K)$ and $g' = K^{-3/2}$ yields:
$$
\int_a^b \frac{O(K)}{K^{3/2}} \, dK = \left[-\frac{2O(K)}{\sqrt{K}}\right]_a^b + 2\int_a^b \frac{O'(K)}{\sqrt{K}} \, dK,
$$
where the derivative $O'(K)$ is $-\Phi(d)$ for call options and $1 - \Phi(d)$ for put options. Consequently, the computational problem reduces to evaluating the auxiliary integral $\int \frac{\Phi(d(K))}{\sqrt{K}}dK$.

\subsection{Model-Specific Analytic Properties}
The feasibility of a closed-form solution depends on the underlying price dynamics:
\begin{itemize}
    \item \textit{Black-Scholes (Closed Form)}: In the lognormal model, $$d_2(K) = \frac{\log(P_0/K) - \tfrac{1}{2}\sigma^2 T}{\sigma\sqrt{T}}$$ is proportional to $\log K$. By substituting $u = \sqrt{K}$, $$d_2(u^2) = \frac{\log P_0 - 2\log u}{\sigma\sqrt{T}} - \frac{\sigma\sqrt{T}}{2} = A - B\log u$$ becomes linear in $\log u$. The resulting integral takes the form $\int \Phi(\text{linear}) \times \exp(\text{linear})$, which admits an elementary antiderivative involving only $\Phi$ and $\exp$.
    \item \textit{Bachelier (Non-Elementary)}: In the normal model, $$d(K) = \frac{F - K}{\sigma\sqrt{T}}$$ is proportional to $K$. The substitution $u = \sqrt{K}$ results in $$d(u^2) = \frac{F - u^2}{\sigma\sqrt{T}} = A - Bu^2,$$ which is quadratic in $u$. The auxiliary integral $\int \Phi(\text{quadratic}) \, du$ has no elementary antiderivative and typically requires Owen’s $T$-functions or complex error functions for exact evaluation.
\end{itemize}

\subsection{Semi-Closed Form Solution}
To maintain high precision without the complexity of non-elementary functions, we compute the Bachelier integral using a semi-closed form:
$$
\int_a^b \frac{C_{\mathrm{B}}(K)}{K^{3/2}} \, dK = \underbrace{\left[-\frac{2C_{\mathrm{B}}(K)}{\sqrt{K}}\right]_a^b}_{\text{exact}} - 2 \underbrace{\int_a^b \frac{\Phi(d)}{\sqrt{K}} \, dK}_{\text{Gauss-Legendre}}
$$
The boundary terms are evaluated exactly. For the remaining integral, we employ $n$-point Gauss-Legendre quadrature on the interval $[a, b]$:
$$
\int_a^b f(K) \, dK \approx \frac{b-a}{2} \sum_{i=1}^{n} w_i \, f\!\left(\frac{b-a}{2}x_i + \frac{a+b}{2}\right),
$$
where $x_i$ and $w_i$ denote the Legendre nodes and weights on $[-1, 1]$.

\subsection{Error Control and Validation}
The integrand $\Phi(d(K))/\sqrt{K}$ is $C^\infty$ on any interval $[a, b] \subset (0, \infty)$, ensuring exponential convergence for Gauss-Legendre quadrature. Empirical testing demonstrates the following precision levels:
\begin{center}
\begin{tabular}{cc}
\hline
$n$ (Points) & Relative Error \\
\hline
8  & $6 \times 10^{-7}$ \\
16 & $2 \times 10^{-15}$ \\
32 & $< 10^{-15}$ \\
\hline
\end{tabular}
\vspace{6pt}
\captionof{table}{Convergence of relative error with increasing number of quadrature points.}
\label{tab:quadrature}
\end{center}
In our implementation, we utilize $n = 32$ points to consistently achieve machine precision with negligible computational overhead.

\section{Discretization of Integral Remainder}
This section provides a formal discretization of the remainder term $I(P_0, T)$, which is essential for practitioners implementing IL pricing within the piecewise-constant liquidity environment of Uniswap V3.

\begin{proposition} \label{prop:discretization}
Assume a constant risk-free rate $r$ and dividend rate $\delta$. The total impermanent loss value is given by:
$$
\Pi_{\text{IL}}(P_0) = \int_{0}^{\infty} L(q) C(q,T)dq + I(P_0, T), 
$$
where, defining $\alpha = \sqrt{1.0001}$ and $\ell_i$ as the liquidity present in the interval $[x_i, x_{i+1})$ with $x_i = (1.0001)^i = \alpha^{2i}$, the remainder $I(P_0, T)$ evaluates to:
\begin{align}
I(P_0, T) =& (\alpha-1)\sum_{i=i_{\min}}^{k-1} \ell_i\Big[ e^{-rT}\alpha^i - P_0 e^{-\delta T}\alpha^{-(i+1)} \Big] \notag \\
&+ \ell_k\Big[ e^{-rT}\big(\sqrt{P_0}-\alpha^k\big) + P_0 e^{-\delta T} \Big(\frac{1}{\sqrt{P_0}}-\alpha^{-k}\Big) \Big]. \label{eqn:remainder_sum}
\end{align}
Here, $i_{\min}$ is the first index with non-zero liquidity, and $k$ is the index such that $x_k \leq P_0 < x_{k+1}$.
\end{proposition}

\begin{corollary}
If a unit of liquidity $\ell(q) = 1$ is deposited in the interval $[a, b)$ where $a = (1.0001)^i$ and $b = (1.0001)^j$ for some integers $i, j$, and $a \leq P_0 \leq b$, then the IL value simplifies to:
$$
\Pi_{\text{IL}}(P_0) = \int_a^b L(q) C(q,T)dq + e^{-rT} \big(\sqrt{P_0} - \sqrt{a}\big) + P_0 e^{-\delta t} \left(\frac{1}{\sqrt{P_0}} - \frac{1}{\sqrt{a}}\right).
$$
\end{corollary}

\subsection{C.1. Proof of Proposition}
Using put-call parity, the impermanent loss integral is split as follows:
$$
\Pi_{\text{IL}}(P_0) = \int_{0}^{\infty} L(q) C(q,T)dq + \underbrace{\int_{0}^{P_0} L(q)\big(qe^{-rT} - P_0 e^{-\delta T}\big)dq}_{=: I(P_0)}.
$$
Since the intrinsic liquidity $\ell(q)$ is piecewise constant on the intervals $[x_i, x_{i+1})$ defined by the Uniswap V3 tick structure, $I(P_0)$ can be expressed in closed form. For any sub-interval $(a,b) \subset [x_i, x_{i+1})$, the integral evaluates as:
\begin{align*}
I_{[a,b]}(P_0, T) &= \int_a^b \frac{\ell_i}{2q^{3/2}}\big(qe^{-rT} - P_0 e^{-\delta T}\big),dq \\
&= \ell_i\Big[ e^{-rT}\big(\sqrt{b}-\sqrt{a}\big) + P_0 e^{-\delta T}\Big(\frac{1}{\sqrt{b}} - \frac{1}{\sqrt{a}}\Big) \Big].
\end{align*}
By discretizing $I(P_0, T)$ as a sum over the active tick range:
\begin{align*}
I(P_0, T) &= \sum_{i=i_{\min}}^{k-1} I_{[x_i,x_{i+1}]}(P_0, T) + I_{[x_k, P_0]}(P_0, T) \\
&= \sum_{i=i_{\min}}^{k-1}
\ell_i\Big[
e^{-rT}\big(\sqrt{x_{i+1}}-\sqrt{x_i}\big)
+ P_0 e^{-\delta T}
\Big(\frac{1}{\sqrt{x_{i+1}}}-\frac{1}{\sqrt{x_i}}\Big)
\Big] \\
&\quad + \ell_k\Big[
e^{-rT}\big(\sqrt{P_0}-\sqrt{x_k}\big)
+ P_0 e^{-\delta T}
\Big(\frac{1}{\sqrt{P_0}}-\frac{1}{\sqrt{x_k}}\Big)
\Big].
\end{align*}
Substituting $\alpha = \sqrt{1.0001}$ such that $\sqrt{x_i} = \alpha^i$ and $1/\sqrt{x_i} = \alpha^{-i}$, we arrive at the final summation form in \eqref{eqn:remainder_sum}.

\section{Additional Snapshot Overlays} \label{sec:appendix-snapshots}

Main-text figures use the November 17, 2025 snapshot. For completeness, this appendix reports the same multi-resolution overlays for additional weekly snapshots.

\subsection*{30bp Pool}

\begin{center}
    \includegraphics[width=0.98\textwidth]{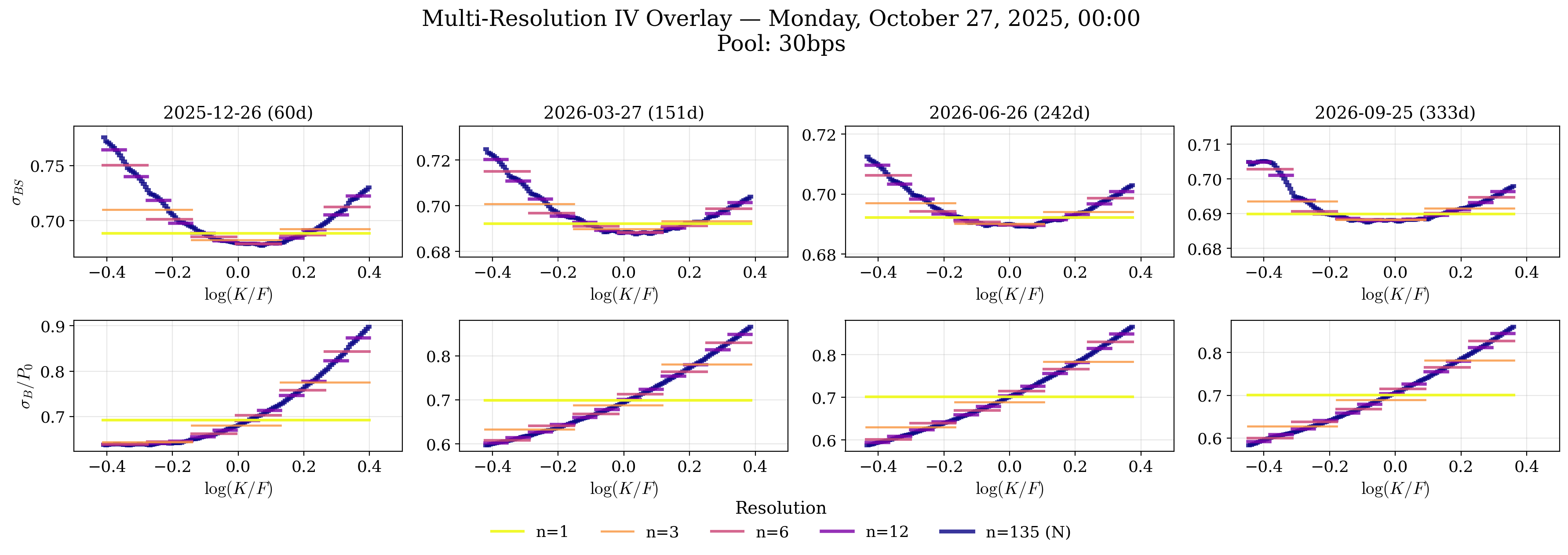}
    \captionof{figure}{Multi-resolution implied-volatility overlay for the 30bp pool (Oct 27, 2025 snapshot).}
\end{center}

\begin{center}
    \includegraphics[width=0.98\textwidth]{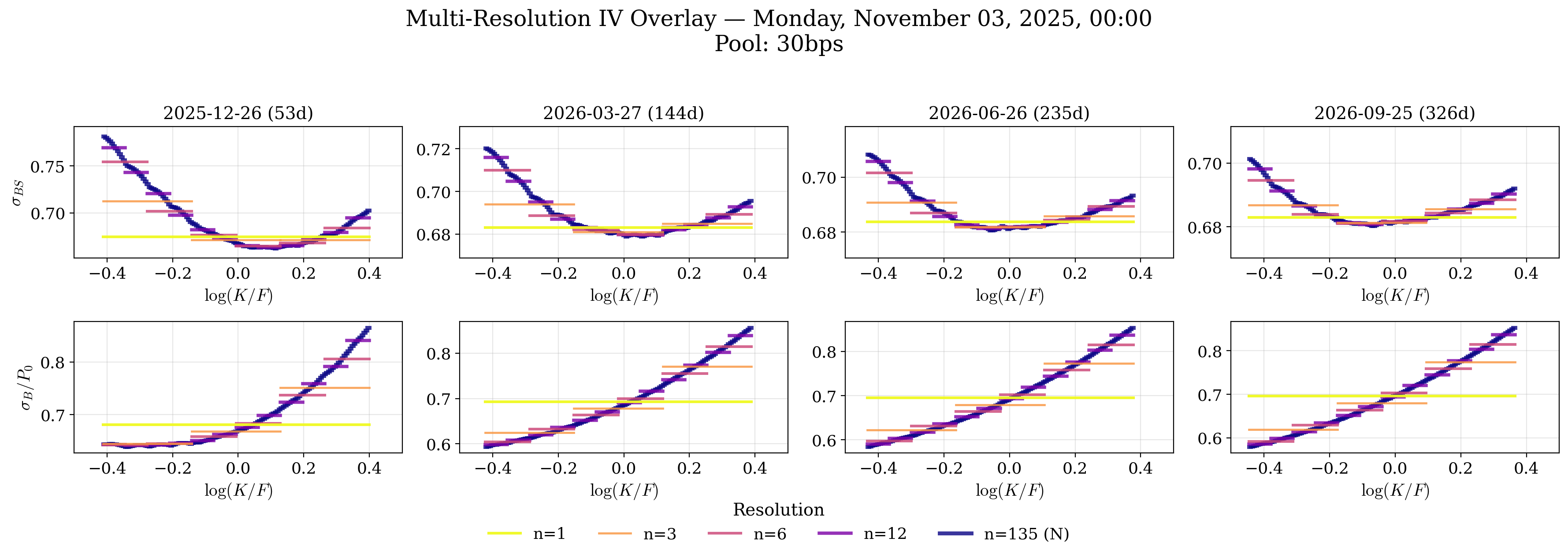}
    \captionof{figure}{Multi-resolution implied-volatility overlay for the 30bp pool (Nov 3, 2025 snapshot).}
\end{center}

\begin{center}
    \includegraphics[width=0.98\textwidth]{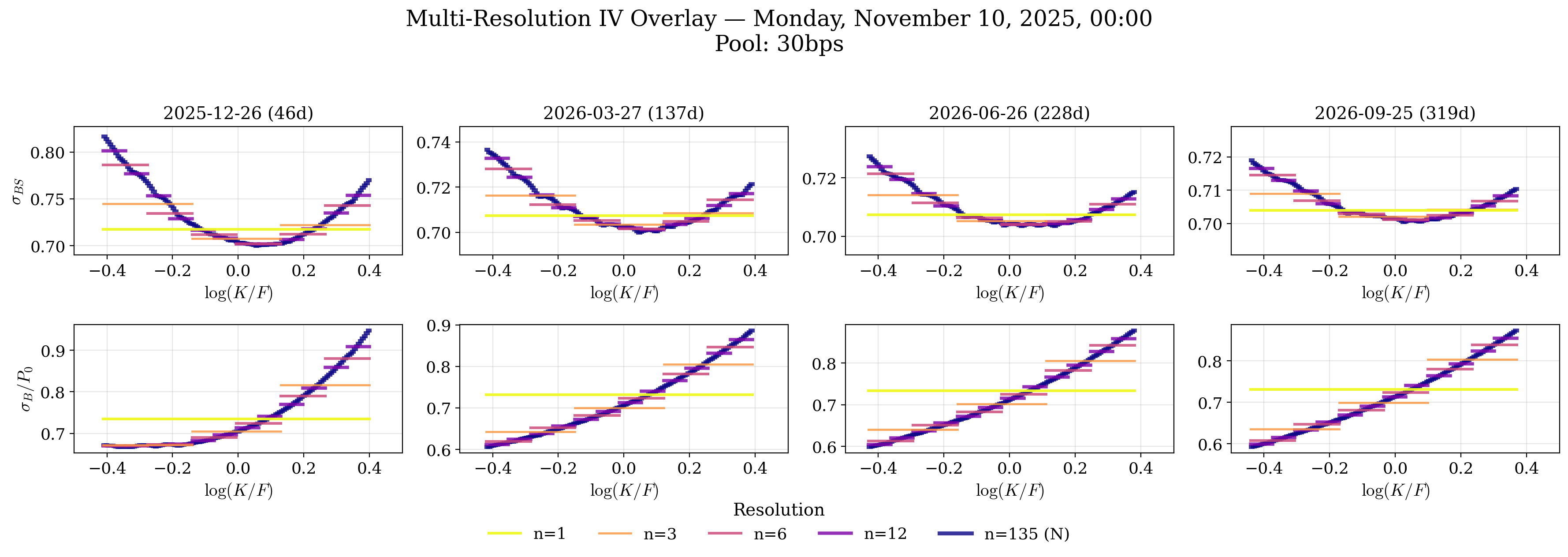}
    \captionof{figure}{Multi-resolution implied-volatility overlay for the 30bp pool (Nov 10, 2025 snapshot).}
\end{center}

\begin{center}
    \includegraphics[width=0.98\textwidth]{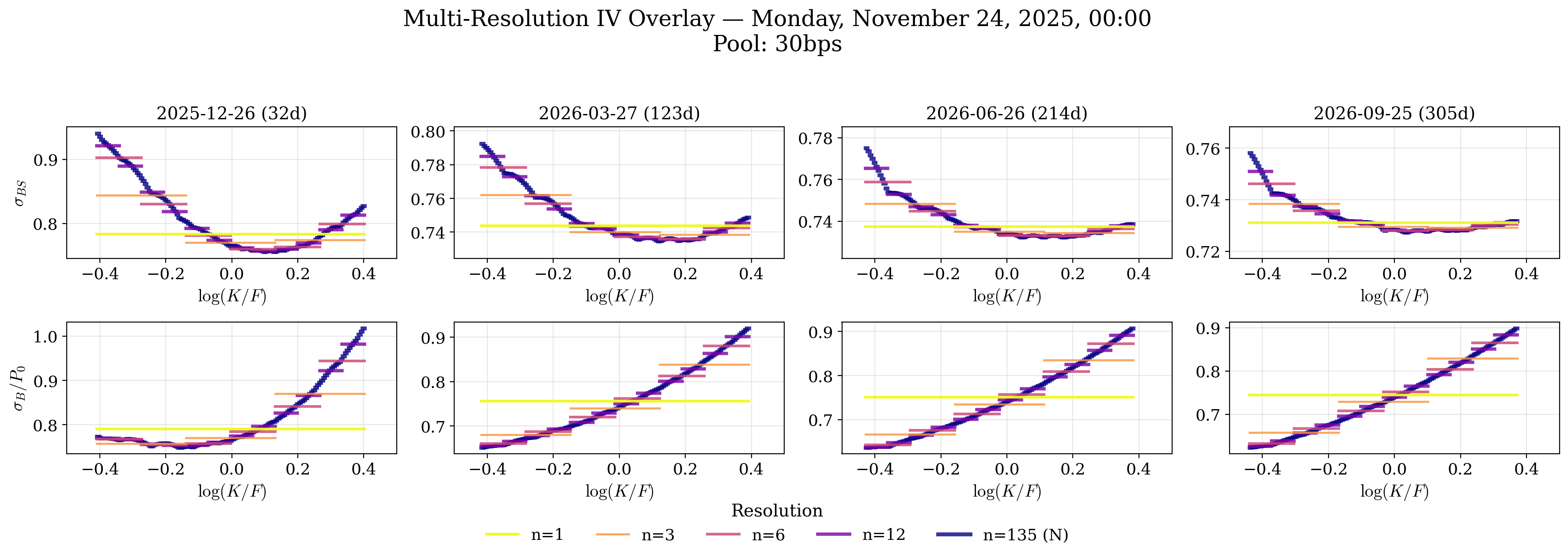}
    \captionof{figure}{Multi-resolution implied-volatility overlay for the 30bp pool (Nov 24, 2025 snapshot).}
\end{center}

\begin{center}
    \includegraphics[width=0.98\textwidth]{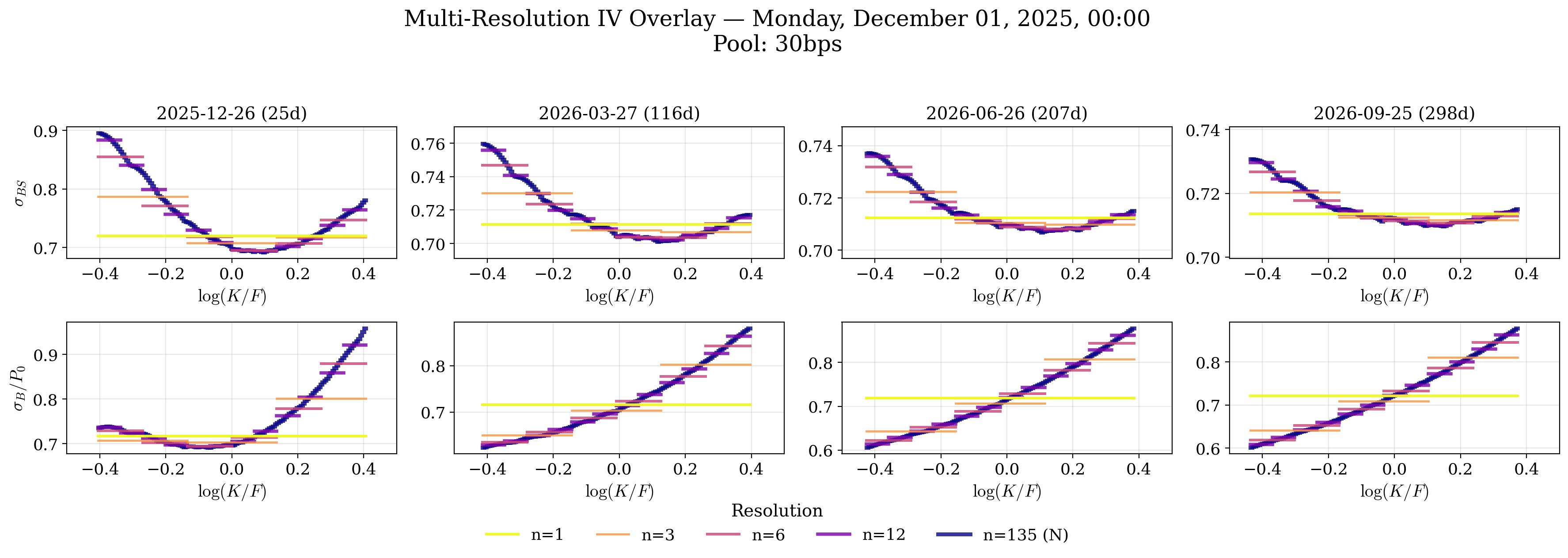}
    \captionof{figure}{Multi-resolution implied-volatility overlay for the 30bp pool (Dec 1, 2025 snapshot).}
\end{center}

\subsection*{5bp Pool}

\begin{center}
    \includegraphics[width=0.98\textwidth]{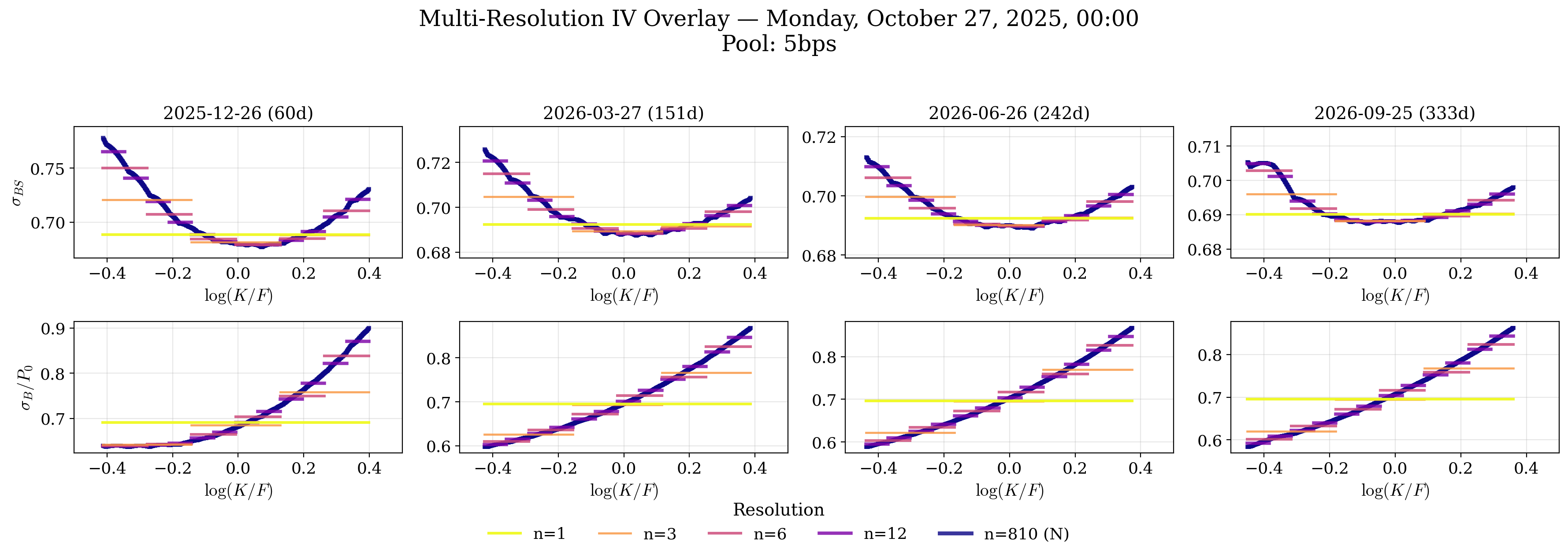}
    \captionof{figure}{Multi-resolution implied-volatility overlay for the 5bp pool (Oct 27, 2025 snapshot).}
\end{center}

\begin{center}
    \includegraphics[width=0.98\textwidth]{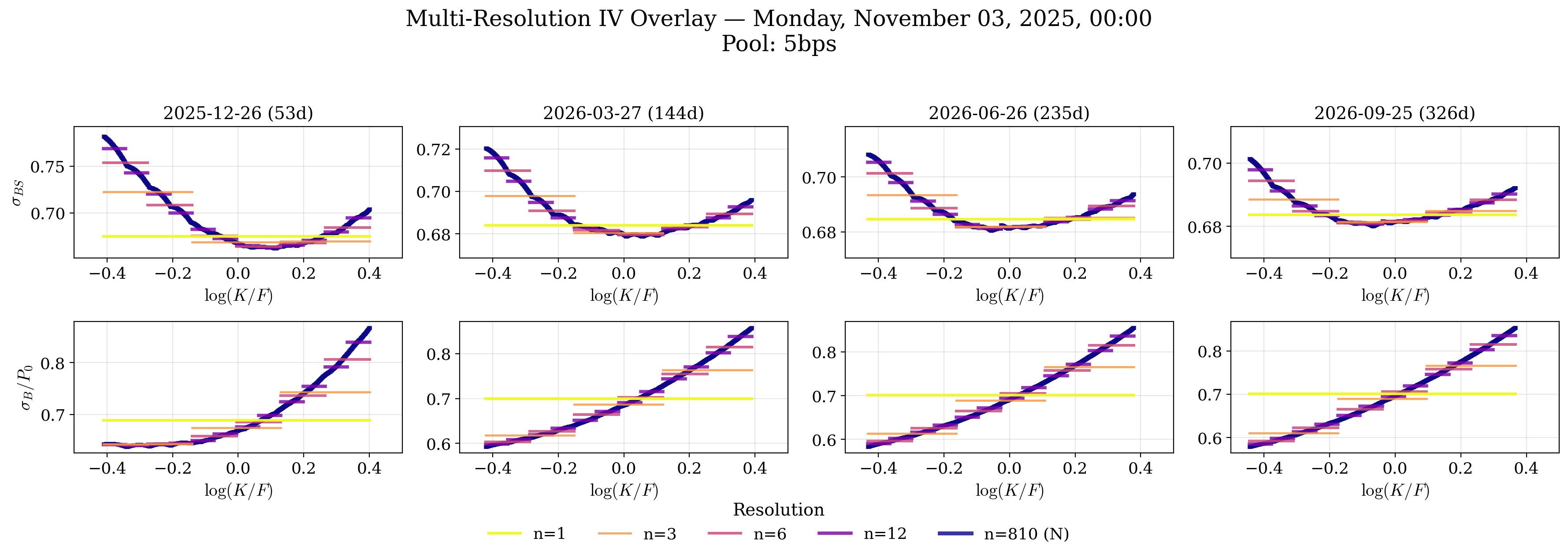}
    \captionof{figure}{Multi-resolution implied-volatility overlay for the 5bp pool (Nov 3, 2025 snapshot).}
\end{center}

\begin{center}
    \includegraphics[width=0.98\textwidth]{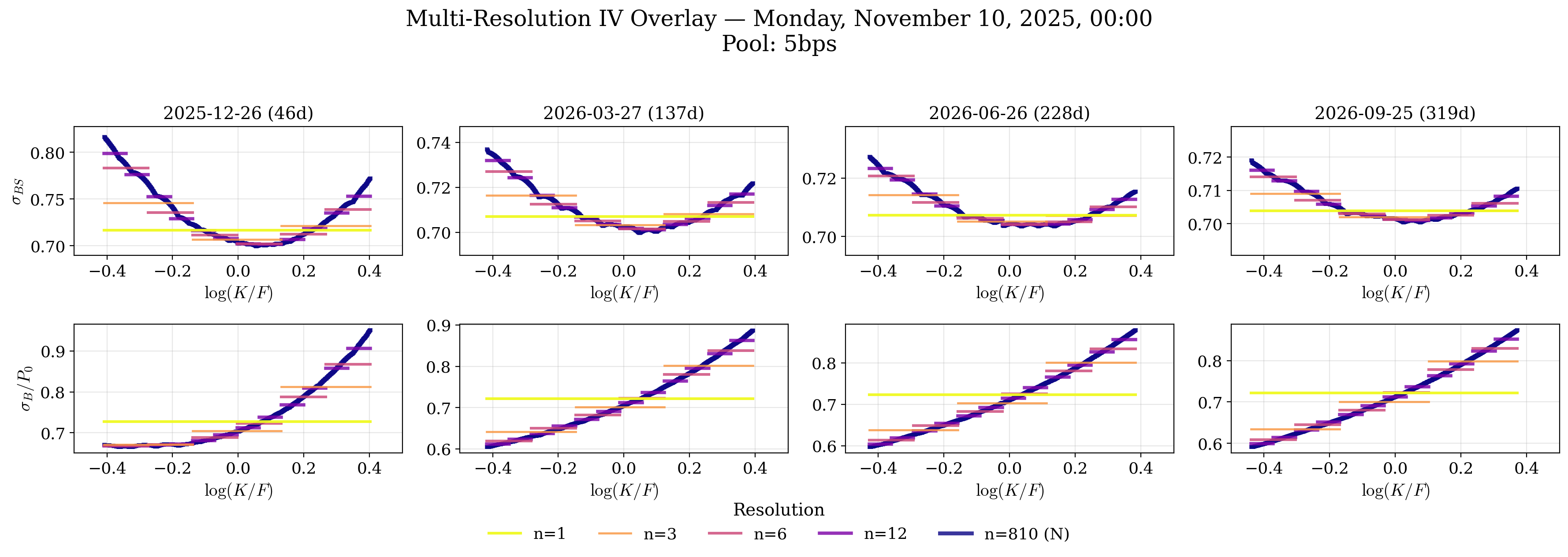}
    \captionof{figure}{Multi-resolution implied-volatility overlay for the 5bp pool (Nov 10, 2025 snapshot).}
\end{center}

\begin{center}
    \includegraphics[width=0.98\textwidth]{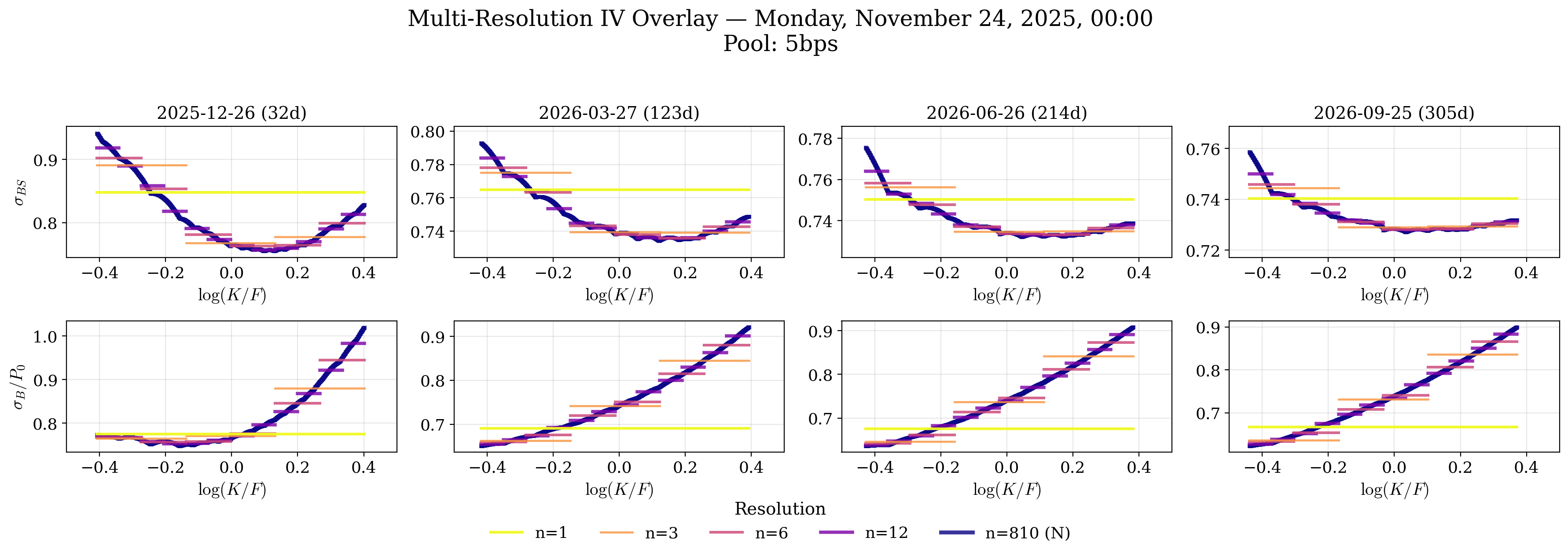}
    \captionof{figure}{Multi-resolution implied-volatility overlay for the 5bp pool (Nov 24, 2025 snapshot).}
\end{center}

\begin{center}
    \includegraphics[width=0.98\textwidth]{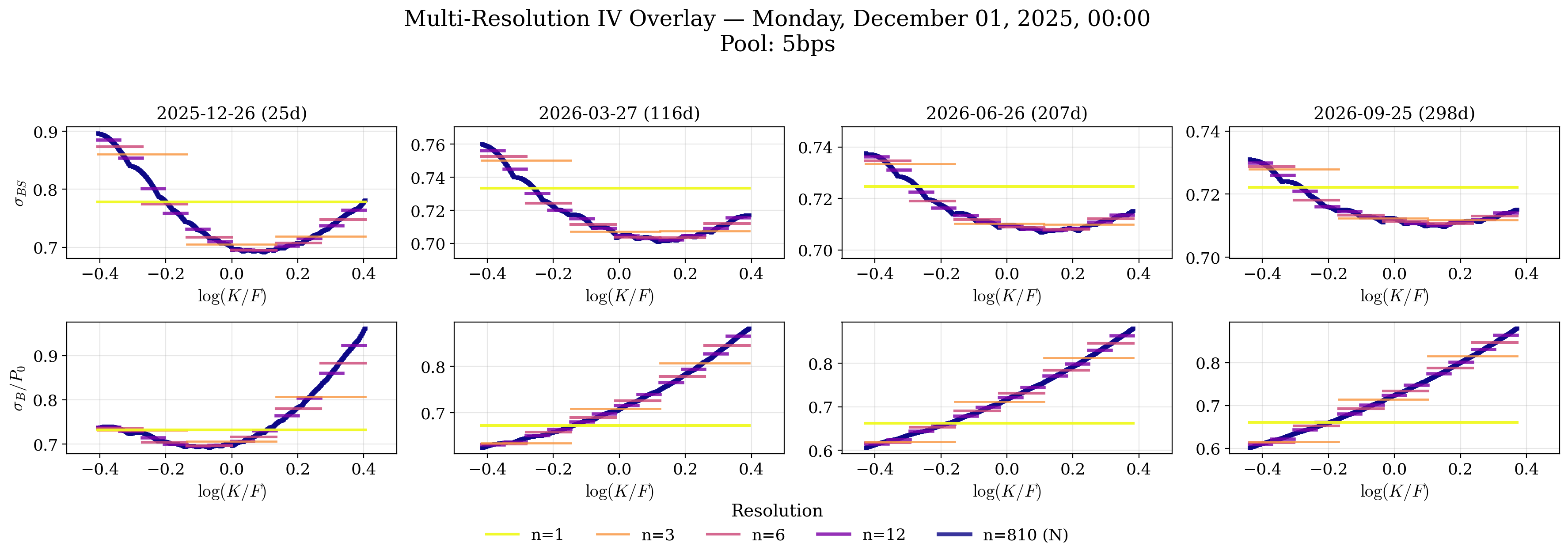}
    \captionof{figure}{Multi-resolution implied-volatility overlay for the 5bp pool (Dec 1, 2025 snapshot).}
\end{center}

\bibliographystyle{alpha}
\bibliography{Ref}

\end{document}